\documentclass[sigconf]{acmart}

\usepackage{amsmath}
\usepackage{amstext}
\usepackage{fp-frame}
\usepackage[latin1]{inputenc}
\usepackage{mathpartir}
\usepackage{fancyvrb}
\usepackage{moreverb}
\usepackage{stmaryrd}
\usepackage{enumerate}
\usepackage{comment}
\usepackage{booktabs,array}
\usepackage{rotating}
\usepackage{xcolor}

\newcommand{\vb}[1]{\verb+#1+}

\newcommand{\defeq}{\triangleq}

\def\squareforqed{\hbox{\rlap{$\sqcap$}$\sqcup$}}
\def\qed{\ifmmode\squareforqed\else{\unskip\nobreak\hfil
\penalty50\hskip1em\null\nobreak\hfil\squareforqed
\parfillskip=0pt\finalhyphendemerits=0\endgraf}\fi}

\newcommand{\ttt}[1]{\texttt{#1}}

\newcommand{\cod}[1]{\llbracket #1 \rrbracket}
\newcommand{\lcod}[2]{\llbracket #1 \rrbracket_{#2}}

\newcommand{\store}{m}
\newcommand{\stores}{\Sigma}
\newcommand{\config}[2]{\langle #1,#2 \rangle}

\newcommand{\extend}[3]{#1\{#2 \mapsto #3\}}

\newcommand{\eassign}[2]{#1 := #2}

\newcounter{topiccounter}
\newcommand{\redx}{\rightarrow}
\newcommand{\redxs}{\redx^*}

\newcommand{\elet}[3]{\mathrm{let}\ #1 = #2\ \mathrm{in}\ #3}

\newcommand{\dom}{\mathrm{dom}}

\newcounter{problemcounter}
\newcounter{solncounter}
\newcommand{\minifed}{\mathit{Overture}}
\newcommand{\minicat}{\minifed}
\newcommand{\fedprot}{\minifed}
\newcommand{\metaprot}{\mathit{Prelude}}

\newcommand{\prog}{\pi}

\newcommand{\bop}{\ \mathit{binop}\ }

\newcommand{\kernel}[2]{#1^{-1}(#2)}

\newcommand{\mems}{\mathit{mems}}

\newcommand{\margd}[2]{{#1}_{#2}}
\newcommand{\condd}[3]{#1_{({#2}|{#3})}}

\newcommand{\progtt}{\mathrm{BD}}
\newcommand{\vars}{\mathit{vars}}
\newcommand{\iov}{\mathit{iovars}}
\newcommand{\flips}{\mathit{flips}}

\newcommand{\fedcat}{\minifed}

\newcommand{\sx}[2]{\elab{\secret{#1}}{#2}}
\newcommand{\mx}[2]{\elab{\mesg{#1}}{#2}} 

\newcommand{\ox}[1]{\out{#1}}

\newcommand{\idealf}{\mathcal{F}}
\newcommand{\SIM}{\mathrm{Sim}}
\newcommand{\prob}{\mathrm{Pr}}
\newcommand{\dist}{\mathrm{D}}

\newcommand{\flip}[2]{\ttt{flip[}#1\ttt{,}#2\ttt{]}}
\newcommand{\secret}[2]{\ttt{s[}#1\ttt{,}#2\ttt{]}}

\newcommand{\enot}{\ttt{not}}
\newcommand{\eand}{\ttt{and}}
\newcommand{\eor}{\ttt{or}}
\newcommand{\exor}{\ttt{xor}}
\renewcommand{\elet}[3]{\ttt{let}\ #1\ \ttt{=}\ #2\ \ttt{in}\ #3}

\renewcommand{\redx}{\xrightarrow{}}
\renewcommand{\redxs}{\xrightarrow{}^{*}}

\newcommand{\secrets}{\mathit{secrets}}

\newcommand{\views}{\mathit{views}}

\newcommand{\cid}{\iota}

\newcommand{\msend}[4]{\elab{\mesg{#1}}{#2}\ \ttt{:=}\ \elab{#3}{#4}}
\newcommand{\OT}[3]{\ttt{OT(} #1 \ttt{,}\ #2 \ttt{,}\ #3 \ttt{)}}

\newcommand{\codebase}{\mathcal{C}}

\newcommand{\flab}{\ell}
\newcommand{\be}{\varepsilon}
\newcommand{\instr}{\mathbf{c}}
\newcommand{\solvealg}{\mathit{models}}
\newcommand{\solve}[3]{\solvealg\ #1\ #2\ #3}

\newcommand{\logit}[1]{\lfloor #1 \rfloor}
\newcommand{\runs}{\mathit{runs}}

\newcommand{\datalog}{\mathit{datalog}}
\newcommand{\concat}{\ttt{++}}

\newcommand{\detx}[1]{\mathbf{D}(#1)}
\newcommand{\unix}[1]{\mathbf{U}(#1)}
\newcommand{\sep}[3]{#1 \vdash #2 * #3}
\newcommand{\condp}[3]{#1|#2 \vdash #3}
\newcommand{\conddetx}[3]{\condp{#1}{#2}{\detx{#3}}}
\newcommand{\condsep}[4]{\condp{#1}{#2}{#3 * #4}}
\newcommand{\condunix}[3]{\condp{#1}{#2}{\unix{#3}}}

\newcommand{\pmf}{\mathit{P}}

\renewcommand{\flip}[1]{\ttt{r[}#1\ttt{]}}
\newcommand{\locflip}{\ttt{r[}\mathtt{local}\ttt{]}}
\renewcommand{\secret}[1]{\ttt{s[}#1\ttt{]}}

\newcommand{\mesg}[1]{\ttt{m[}#1\ttt{]}}
\newcommand{\out}[1]{\elab{\ttt{out}}{#1}}
\newcommand{\rvl}[1]{\ttt{p[}#1\ttt{]}}

\newcommand{\elab}[2]{#1\ttt{@}#2}
\renewcommand{\eassign}[4]{\elab{#1}{#2} := \elab{#3}{#4}}
\newcommand{\xassign}[3]{#1 := \elab{#2}{#3}}
\newcommand{\pubout}[3]{\out{#1} := \elab{#2}{#3}}
\newcommand{\reveal}[3]{\rvl{#1} := \elab{#2}{#3}}

\newcommand{\adversary}{\mathcal{A}}
\newcommand{\aredx}{\redx_{\adversary}}
\newcommand{\aredxs}{\redxs_{\adversary}}
\newcommand{\arewrite}{\mathit{rewrite}_{\adversary}}
\newcommand{\cinputs}{V_{C \rhd H}}
\newcommand{\houtputs}{V_{H \rhd C}}
\newcommand{\aruns}{\mathit{runs}_\adversary}
\newcommand{\botruns}{\mathit{runs}_{\adversary,\bot}}

\renewcommand{\store}{\sigma}

\renewcommand{\runs}{\mathit{runs}}

\newcommand{\fcod}[1]{\lcod{#1}{}}
\renewcommand{\flips}{\mathit{rands}}

\newcommand{\ftimes}{*}
\newcommand{\fplus}{+}
\newcommand{\fminus}{-}
\newcommand{\macgv}[1]{\langle #1 \rangle}

\newcommand{\preproc}{\mathit{preproc}}
\newcommand{\assert}[1]{\ttt{assert(}#1\ttt{)}}
\newcommand{\mv}{\nu}
\newcommand{\andgmw}{\ttt{andgmw}}
\newcommand{\decodegmw}{\ttt{decodegmw}}
\newcommand{\bodies}{\mathit{bodies}}

\newcommand{\compfig}
{
  \begin{fpfig}[t]{Comparison of systems for verification of MPC security in PLs.}{fig-comp}
    \begin{tabular}{lccccccc}
      & 
      \begin{sideways} probabilistic language \end{sideways} &
      \begin{sideways} probabilistic conditioning \end{sideways} & 
      \begin{sideways} low-level protocols \end{sideways} & 
      \begin{sideways} passive security \end{sideways} & 
      \begin{sideways} malicious security \end{sideways}& 
      \begin{sideways} hyperproperties \end{sideways}& 
      \begin{sideways} automation \end{sideways}\\\hline\hline
      Haskell EDSL \cite{6266151} & \checkmark &  & \checkmark  & \checkmark & & \checkmark & \checkmark \\\hline
      MPC in SecreC \cite{almeida2018enforcing} & \checkmark & \checkmark &   & \checkmark & & \checkmark & \checkmark \\\hline
      $\lambda_{\text{obliv}}$ \cite{darais2019language} & \checkmark & & \checkmark & & & \checkmark & \checkmark \\\hline
      PSL \cite{barthe2019probabilistic} & \checkmark & & \checkmark & & & & \\\hline
      Lilac \cite{li2023lilac} & \checkmark & \checkmark & & & & & \\\hline
      Wys$^*$ \cite{wysstar} & & & & \checkmark & & \checkmark & \checkmark\\\hline
      Viaduct \cite{10.1145/3453483.3454074,viaduct-UC} & & & & \checkmark & \checkmark & \checkmark & \checkmark\\\hline
      MPC in EasyCrypt \cite{8429300} &  \checkmark &  \checkmark &  \checkmark & \checkmark & \checkmark & \checkmark & \\\hline
      $\metaprot$/$\minifed$ & \checkmark & \checkmark & \checkmark & \checkmark & \checkmark & \checkmark & \checkmark\\
      \hline
    \end{tabular}
  \end{fpfig}
}

\newcommand{\minifedfig}
{
\begin{fpfig}[t]{Top-to-bottom: Basic $\minifed$ syntax, expression interpretation, and command evaluation.}{fig-minifed}
  {
    $$
    \begin{array}{rcl@{\hspace{2mm}}r}
      \multicolumn{4}{l}{v \in \mathbb{F}_p,\ w \in \mathrm{String},\ \cid \in \mathrm{Clients} \subset  \mathbb{N} }\\[2mm] 
      \be &::=& \flip{w} \mid \secret{w} \mid \mesg{w} \mid \rvl{w} \mid & \textit{expressions}\\
      & & v \mid \be \fminus \be \mid \be \fplus \be \mid \be \ftimes \be \\[2mm]
      x &::=& \elab{\flip{w}}{\cid} \mid \elab{\secret{w}}{\cid} \mid \elab{\mesg{w}}{\cid} \mid \rvl{w} \mid \out{\cid} & \textit{variables} \\[2mm]
      \instr &::=& \eassign{\mesg{w}}{\cid}{\be}{\cid} \mid
      \reveal{w}{e}{\cid} \mid \pubout{\cid}{\be}{\cid} & \textit{commands} \\[2mm]
      \prog &::=& \varnothing \mid \instr; \prog & \textit{protocols}
    \end{array}
    $$

    \medskip
  
  \rule{80mm}{0.5pt}
    
  $$
  \begin{array}{rcl}
    \lcod{\store, v}{\cid} &=& v\\
    \lcod{\store, \be_1 \fplus \be_2}{\cid} &=& \fcod{\lcod{\store, \be_1}{\cid} \fplus \lcod{\store, \be_2}{\cid}}\\ 
    \lcod{\store, \be_1 \fminus \be_2}{\cid} &=& \fcod{\lcod{\store, \be_1}{\cid} \fminus \lcod{\store, \be_2}{\cid}}\\ 
    \lcod{\store, \be_1 \ftimes \be_2}{\cid} &=& \fcod{\lcod{\store, \be_1}{\cid} \ftimes \lcod{\store, \be_2}{\cid}}\\
    \lcod{\store, \flip{w}}{\cid} &=& \store(\elab{\flip{w}}{\cid})\\
    \lcod{\store, \secret{w}}{\cid} &=& \store(\elab{\secret{w}}{\cid})\\
    \lcod{\store, \mesg{w}}{\cid} &=& \store(\elab{\mesg{w}}{\cid})\\
    \lcod{\store, \rvl{w}}{\cid} &=& \store(\rvl{w})\\
  \end{array}
  $$

  \vspace{4mm}
  
  \rule{80mm}{0.5pt}

  \begin{mathpar}
    (\store, \xassign{x}{\be}{\cid};\prog) \redx (\extend{\store}{x}{\lcod{\store,\be}{\cid}}, \prog)
  \end{mathpar}
  }
\end{fpfig}
}

\newcommand{\adversaryfig}
{
\begin{fpfig}[h]{Adversarial semantics, and semantics of $\ttt{assert}$.}{fig-adversary}
  {
    $$
    \begin{array}{rclr}
      (\store, \xassign{x}{\be}{\cid};\prog) &\aredx&
      (\extend{\store}{x}{\lcod{\store,\be}{\cid}}, \prog) & \cid \in H\\
      (\store, \xassign{x}{\be}{\cid};\prog) &\aredx&
      (\extend{\store}{x}{\lcod{\arewrite(\store_C,\be)}{\cid}}, \prog) & \cid \in C
    \end{array}
    $$
    
    $$
    \begin{array}{rcl@{\qquad}r}
      (\store,\elab{\assert{\phi(\be)}}{\cid};\prog) &\aredx& (\store,\prog) & \text{if\ } \phi(\lcod{\store,\be}{\cid}) \text{\ or\ } \cid \in C\\
      (\store,\elab{\assert{\phi(\be)}}{\cid};\prog) &\aredx& (\store,\varnothing)  & \text{if\ } \neg\phi(\store,\lcod{\store,\be}{\cid})
    \end{array}
    $$
  }
\end{fpfig}
}

\newcommand{\solvefig}
{
  \begin{fpfig}[t]{Filtering solutions to expressions in $\mathbb{F}_2$.}{fig-solve}
    {\small \begin{eqnarray*}
        \solve{\stores}{1}{\cid} &=& \stores\\
        \solve{\stores}{0}{\cid} &=& \varnothing\\
        \solve{\stores}{\mesg{w}}{\cid} &=& \{ \store \in \stores \mid \store(\elab{\mesg{w}}{\cid}) = 1 \} \\
        \solve{\stores}{\flip{w}}{\cid} &=& \{ \store \in \stores \mid \store(\elab{\flip{w}}{\cid}) = 1 \} \\
        \solve{\stores}{\secret{w}}{\cid} &=& \{ \store \in \stores \mid \store(\elab{\secret{w}}{\cid}) = 1 \} \\
        \solve{\stores}{\rvl{w}}{\cid} &=& \{ \store \in \stores \mid \store(\rvl{w}) = 1 \} \\
        \solve{\stores}{(\enot\ \be)}{\cid} &=& \stores - (\solve{\stores}{\be}{\cid})\\
        \solve{\stores}{(\be_1\ \eand\ \be_2)}{\cid} &=& (\solve{\stores}{\be_1}{\cid}) \cap (\solve{\stores}{\be_2}{\cid}) \\
        \solve{\stores}{(\be_1\ \eor\ \be_2)}{\cid} &=& (\solve{\stores}{\be_1}{\cid}) \cup (\solve{\stores}{\be_2}{\cid}) \\
        \solve{\stores}{(\be_1\ \exor\ \be_2)}{\cid} &=&
        ((\solve{\stores}{\be_1}{\cid}) \cap (\stores - \solve{\stores}{\be_2}{\cid})) \cup\\
        && ((\stores - \solve{\stores}{\be_1}{\cid}) \cap (\solve{\stores}{\be_2}{\cid})) 
      \end{eqnarray*} }
  \end{fpfig}
}

\newcommand{\metaprotfig}
{
  \begin{fpfig}[t]{$\metaprot$ syntax (T), evaluation contexts (M), and operational semantics (B).}{fig-metaprot}
    {\small{
    $$
    \begin{array}{rcl}
      \multicolumn{3}{l}{\flab \in \mathrm{Field},\   y \in \mathrm{EVar}, \  f \in \mathrm{FName}}\\[1mm]
      e &::=& \mv \mid \flip{e} \mid \secret{e} \mid \mesg{e} \mid \rvl{e} \mid e \bop e \mid \\
      & & \msend{e}{e}{e}{e} \mid \reveal{e}{e}{e} \mid \pubout{e}{e}{e} \mid \\
      & &  \elab{\assert{\phi(e)}}{e} \mid y \mid \elet{y}{e}{e} \mid \\
      & & f(e,\ldots,e) \mid \{ \flab = e; \ldots; \flab = e \} \mid e.\flab \mid e;e \\[1mm]
      \bop &::=& \fplus \mid \fminus \mid \ftimes \mid \concat  \\[1mm]
      \mv &::=& w \mid \cid \mid \be \mid \{ \flab = \mv;\ldots;\flab = \mv \} 
      \mid \ttt{()} \\[1mm] 
           \mathit{fn} &::=& f(y,\ldots,y) \{ e \} 
    \end{array}
    $$
    }

    \rule{80mm}{0.5pt}

    {\small
    $$
      \begin{array}{c}
        E \ ::= \\
        \begin{array}{l}
          [\,] \mid E\ \bop\ e \mid \mv\ \bop\ E \mid \flip{E} \mid \secret{E} \mid \mesg{E} \mid \rvl{E} \mid \\
          \msend{E}{e}{e}{e} \mid \msend{\mv}{E}{e}{e} \mid \msend{\mv}{\mv}{E}{e} \mid \msend{\mv}{\mv}{\mv}{E} \mid\\
          \reveal{E}{e}{e} \mid \reveal{\mv}{E}{e} \mid \reveal{\mv}{\mv}{E} \mid\\
          \pubout{E}{e}{e} \mid \pubout{\mv}{E}{e} \mid \pubout{\mv}{\mv}{E} \mid\\
          \elab{\assert{\phi(E)}}{e} \mid \elab{\assert{\phi(\mv)}}{E} \mid\\
          \elet{y}{E}{e} \mid f(\mv,\ldots,\mv,E,e,\ldots,e) \mid E;e \\
          \{ \flab = \mv;\ldots;\flab = \mv;\flab = E;\flab = e;\ldots;\flab = e \} \mid E.\flab 
        \end{array}
      \end{array}
    $$
    \vspace{.5mm}
    
    \rule{80mm}{0.5pt}
    
    \begin{mathpar}
      \config{\prog}{\elet{y}{\mv}{e}} \redx \config{\prog}{e[\mv/y]}
       
      \inferrule
          {\codebase(f) = y_1,\ldots,y_n,\ e}
          {\config{\prog}{f(\mv_1,...,\mv_n)} \redx \config{\prog}{e[\mv_1/y_1,\ldots,\mv_n/y_n]}} 
      
      \config{\prog}{\{\ldots; \flab = \mv; \ldots\}.\flab} \redx \config{\prog}{\mv}
      
      \config{\prog}{w_1\concat w_2} \redx \config{\prog}{w_1w_2}
      
      \config{\prog}{\mv;e} \redx \config{\prog}{e}
      
      \config{\prog}{\instr} \redx \config{\prog;\instr}{()}
      
      \inferrule
          {\config{\prog}{e} \redx \config{\prog'}{e'}}
          {\config{\prog}{E[e]} \redx \config{\prog'}{E[e']}}
    \end{mathpar}}
  }
\end{fpfig}
}

\setcopyright{acmlicensed}  %
\copyrightyear{2024}
\acmYear{2024}
\acmConference[PPDP '24]{PPDP '24}{September 10-11, 2024}{Milano, Italy}
\acmBooktitle{Proceedings of the 26th Symposium on Principles and Practice of Declarative Programming, PPDP 2024, Milano, Italy, September 10-11, 2024}
\acmPrice{15.00}
\acmISBN{9-8-4007-0969-297}

\begin{document}

\title{Language-Based Security for Low-Level MPC}

\author{Christian Skalka}
\affiliation{
  \institution{University of Vermont}
  \city{}
  \country{}
}
\email{ceskalka@uvm.edu}

\author{Joseph P. Near}
\affiliation{
  \institution{University of Vermont}
  \city{}
  \country{}
}
\email{jnear@uvm.edu}

\begin{abstract}
  Secure Multi-Party Computation (MPC) is an important
  enabling technology for data privacy in modern distributed
  applications. Currently, proof methods for low-level MPC protocols
  are primarily manual and thus tedious and error-prone, and are also
  non-standardized and unfamiliar to most PL theorists. As a step
  towards better language support and language-based enforcement, we
  develop a new staged PL for defining a variety of low-level
  probabilistic MPC protocols. We also formulate a collection of
  confidentiality and integrity hyperproperties for our language model
  that are familiar from information flow, including conditional
  noninterference, gradual release, and robust declassification. We
  demonstrate their relation to standard MPC threat models of passive
  and malicious security, and how they can be leveraged in security
  verification of protocols. To prove these properties we develop
  automated tactics in $\mathbb{F}_2$ that can be integrated with
  separation logic-style reasoning.
\end{abstract}



\settopmatter{printfolios=true}
\maketitle

\section{Introduction}

Secure Multi-Party Computation (MPC) protocols support data privacy in
important modern, distributed applications such as privacy-preserving
machine learning \cite{li2021privacy, knott2021crypten,
  koch2020privacy, liu2020privacy} and Zero-Knowledge proofs in
blockchains \cite{ishai2009zero, lu2019honeybadgermpc,
  gao2022symmeproof, tomaz2020preserving}. The security semantics of
MPC include both confidentiality and integrity properties incorporated
into models such as real/ideal (aka simulator) security and universal
composability (UC), developed primarily by the cryptography community
\cite{evans2018pragmatic}.  Related proof methods are well-studied
\cite{Lindell2017} but mostly manual. Somewhat independently, a
significant body of work in programming languages has focused on
definition and enforcement of confidentiality and integrity
\emph{hyperproperties} \cite{10.5555/1891823.1891830} such as
noninterference and gradual release
\cite{4223226,sabelfeld2009declassification}. Following a tradition of
connecting cryptographic and PL-based security models
\cite{10.1007/3-540-44929-9_1,10.1145/3571740}, recent work has also
recognized connections between MPC security models and hyperproperties
of, e.g., noninterference \cite{8429300}, and even leveraged these
connections to enforce MPC security through mechanisms such as
security types \cite{10.1145/3453483.3454074}. Major benefits of this
connection in an MPC setting include better language abstractions for
defining protocols and for mechanization and even automation of
security proofs.  The goal of this paper is to develop a PL model for
defining a variety of low-level probabilistic MPC protocols, to
formulate a collection of confidentiality and integrity
hyperproperties for our model with familiar information flow
analogs, and to show how these properties can be leveraged for
improved proof automation.

The distinction between high- and low-level languages for MPC is
important. High-level languages such as Wysteria
\cite{rastogi2014wysteria} and Viaduct \cite{10.1145/3453483.3454074}
are designed to provide effective programming of full
applications. These language designs incorporate sophisticated
verified compilation techniques such as orchestration
\cite{viaduct-UC} to guarantee high-level security properties, and
they rely on \emph{libraries} of low-level MPC protocols, such as
binary and arithmetic circuits. These low-level protocols encapsulate
abstractions such as secret sharing and semi-homomorphic encryption,
and must be verified by hand. So, low-level MPC
programming and protocol verification remains a distinct challenge and both
critical to the general challenge of PL design for MPC and
complementary to high-level language design.

The connection between information flow hyperproperties and MPC
security is also complicated especially at a low level.  MPC protocols
involve communication between a group of distributed clients called a
\emph{federation} that collaboratively compute and publish the result
of some known \emph{ideal functionality} $\idealf$, maintaining
confidentiality of inputs to $\idealf$ without the use of a trusted
third party. However, since the outputs of $\idealf$ are public, some
information about inputs is inevitably leaked. Thus, the ideal
functionality establishes a declassification policy
\cite{sabelfeld2009declassification}, which is more difficult to
enforce than pure noninterference.  And subtleties of, e.g.,
semi-homomorphic encryption are central to both confidentiality and
integrity properties of protocols and similarly difficult to track
with coarse-grained security types alone.

Nevertheless, as previous authors have observed
\cite{5a51987acaa84c43bb4bf5bcc7d01683}, low-level protocol design
patterns such as secret sharing and circuit gate structure have
compositional properties that can be independently verified and then
leveraged in larger proof contexts. We contribute to this line of work
by developing an automated verification technique for subprotocols and
show how it can be integrated as a tactic in a larger security proof.

\subsection{Overview and Contributions}

In summary, our work provides a complete methodology for end-to-end verification
of MPC protocols via three components:
\begin{enumerate}
\item A \textbf{low-level language} for defining MPC primitives
  (Section \ref{section-lang}) with an associated \textbf{metalanguage}
  to ease programming (Section \ref{section-metalang})\footnote{By metalanguage
  we mean a multi-stage aka metaprogramming language where code is a value, as
  in, e.g., MetaML \cite{TAHA2000211}.}. 
\item A \textbf{fully-automated verification method} for low-level MPC
  primitives in $\mathbb{F}_2$ (Section \ref{section-bruteforce}).
\item A \textbf{partially-automated verification method} for
  MPC protocols, which leverages automated proofs for
  low-level primitives (Section \ref{section-example-gmw}).
\end{enumerate}
As part of this methodology, we also develop hyperproperties that
encode MPC security that may be of independent interest. The complete
methodology enables the verification of real-world MPC protocols like
GMW~\cite{goldreich2019play}. 

\paragraph{Language design.} In Section \ref{section-lang} we
develop a new probabilistic programming language $\minifed$ for
defining synchronous distributed protocols over an arbitrary
arithmetic field. The syntax and semantics provides a succinct account
of synchronous messaging between protocol \emph{clients}. In Section
\ref{section-metalang} we define a metalanguage $\metaprot$ that
dynamically generates $\minifed$ protocols. It is able to express
important low-level abstractions, as we illustrate via implementations
of protocols including Shamir addition (Section \ref{section-lang}),
GMW boolean circuits (Section \ref{section-example-gmw}), and Beaver
Triple multiplication gates with BDOZ authentication (Section
\ref{section-example-bdoz}).

\paragraph{Hyperproperty formulation.} In Section \ref{section-model} we
develop our formalism for expressing the joint probability mass function of
program variables, and give standard definitions of passive and
malicious real/ideal security in our model. In Section
\ref{section-hyper}, we formulate a variety of familiar information
flow properties in our probabilistic setting, including conditional
noninterference, gradual release, and robust declassification, and
consider the relation between these and real/ideal security.  While it
has been previously shown that probabilistic conditional
noninterference is sound for passive security, we formulate new
properties of integrity which, paired with passive security, imply
malicious security (Theorem \ref{theorem-integrity}). We observe
in Section \ref{section-example-bdoz} that authentication mechanisms
such as BDOZ/SPDZ style MACs enforce a strictly weaker property
of ``cheating detection'' (Lemma \ref{lemma-cheating}).

\paragraph{Fully and partially automated verification.} In Section
\ref{section-bruteforce} we develop a method for automatically
computing the probability mass function (pmf) of $\minifed$ protocols
in $\mathbb{F}_2$, that can be automatically queried to enforce
hyperproperties of security. This method is perfectly accurate but has
high complexity; we show this can be partially mitigated by conversion
of protocols in $\mathbb{F}_2$ to stratified Datalog which is amenable
to HPC acceleration. Furthermore, in Section \ref{section-example-gmw}
we consider in detail how this automated technique can be used as a
local automated tactic for proving security in arbitrarily large GMW
circuits using conditional probabilistic independence as in
\cite{li2023lilac} (Lemmas \ref{lemma-gmwtactic} and
\ref{lemma-gmwinvariant} and Theorem \ref{theorem-gmw}).

\compfig

\subsection{Related Work}
\label{section-related-work}

Our main focus is on PL design and automated and semi-automated
reasoning about security properties of low-level MPC protocols. Prior
work has considered \textbf{automated} verification of
\textbf{high-level} protocols and \textbf{manual} verification of
\textbf{low-level} protocols---but none offers the combination of
automation and low-level support we consider.
We summarize this comparison in Figure
\ref{fig-comp}, with the caveat that works vary in the degree of
development in each dimension.


As mentioned above, several high-level languages have been developed
for writing MPC applications, and frequently exploit the connection
between hyperproperties and MPC security. Previous work on analysis
for the SecreC language
\cite{almeida2018enforcing,10.1145/2637113.2637119} is concerned with
properties of complex MPC circuits, in particular a user-friendly
specification and automated enforcement of declassification bounds in
programs that use MPC in subprograms. The Wys$^\star$ language
\cite{wysstar}, based on Wysteria \cite{rastogi2014wysteria}, has
similar goals and includes a trace-based semantics for reasoning about
the interactions of MPC protocols. Their compiler also guarantees that
underlying multi-threaded protocols enforce the single-threaded source
language semantics. These two lines of work were focused on passive
security. The Viaduct language
\cite{10.1145/3453483.3454074} has a well-developed
information flow type system that automatically enforces both
confidentiality and integrity through hyperproperties such as robust
declassification, in addition to rigorous compilation guarantees
through orchestration \cite{viaduct-UC}. However, these high level
languages lack probabilistic features and other abstractions of
low-level protocols, the implementation and security of which are
typically assumed as a selection of library components.

Various related low-level languages with probabilistic features have
also been developed. The $\lambda_{\mathrm{obliv}}$ language
\cite{darais2019language} uses a type system to automatically
enforce so-called probabilistic trace obliviousness.  But similar to
previous work on oblivious data structures \cite{10.1145/3498713},
obliviousness is related to pure noninterference, not the relaxed form
related to passive MPC security. The Haskell-based security type
system in \cite{6266151} enforces a version of noninterference that is
sound for passive security, but does not verify the correctness of
declassifications and does not consider malicious security. And
properties of real/ideal passive and malicious security for a
probabilistic language have been formulated in EasyCrypt
\cite{8429300}-- though their proof methods, while mechanized, are
fully manual, and their formulation of malicious security is not as
clearly related to robust declassification as is the one we present in
Section \ref{section-hyper}.

Program logics for probabilistic languages and specifically reasoning
about properties such as joint probabilistic independence is also
important related work. Probabilistic Separation Logic (PSL)
\cite{barthe2019probabilistic} develops a logical framework for
reasoning about probabilistic independence (aka separation) in
programs, and they consider several (hyper)properties, such as perfect
secrecy of one-time-pads and indistinguishability in secret sharing,
that are critical to MPC. However, their methods are manual, and
don't include conditional independence (separation). This
latter issue has been addressed in Lilac \cite{li2023lilac}. The
application of Lilac-style reasoning to MPC protocols has not
previously been explored, as we do in Section
\ref{section-example-gmw}.

Our work also shares many ideas with probabilistic programming
languages designed to perform (exact or approximate) statistical
inference~\cite{holtzen2020scaling, carpenter2017stan, wood2014new,
  bingham2019pyro, albarghouthi2017fairsquare, de2007problog,
  pfeffer2009figaro, saad2021sppl}. Our setting, however, requires
verifying properties beyond inference, including conditional
statistical independence. Recent work by Li et al.~\cite{li2023lilac} proposes a
manual approach for proving such properties, but does not provide
automation.

\section{The $\minicat$ Protocol Language}
\label{section-lang}

The $\minifed$ language establishes a basic model of synchronous
protocols between a federation of \emph{clients} exchanging values in
the binary field. A model of synchronous communication captures a wide
range of MPC protocols. Concurrency is out of scope in this work but
an avenue for future work. The lack of sophisticated control
structures in $\minifed$ is intentional, since minimizing features
eases analysis and control abstractions such as function definitions
can be integrated into a metalanguage that generates $\minifed$
programs (Section \ref{section-metalang}).

We identify clients by natural numbers and federations- finite sets of
clients- are always given statically.  Our threat model assumes a
partition of the federation into \emph{honest} $H$ and \emph{corrupt}
$C$ subsets. We model probabilistic programming via a \emph{random
tape} semantics. That is, we will assume that programs can make
reference to values chosen from a uniform random distributions defined
in the initial program memory.  Programs aka protocols execute
deterministically given the random tape.

\subsection{Syntax}

\minifedfig

The syntax of $\minifed$, defined in Figure \ref{fig-minifed},
includes values $v$ and standard operations of addition, subtraction,
and multiplication in a finite field $\mathbb{F}_p$ where $p$ is some
prime.  Protocols are given input secret values $\secret{w}$ as well
as random samples $\flip{w}$ on the input tape, implemented using a
\emph{memory} as described below (Section
\ref{section-lang-semantics}) where $w$ is a distinguishing 
identifier string. Protocols are sequences of assignment commands of three
different forms:
\begin{itemize}
\item $\eassign{\mesg{w}}{\cid_2}{\be}{\cid_1}$: This
  is a \emph{message send} where expression $\be$ is computed
  by client $\cid_1$ and sent to client $\cid_2$ as message
  $\mesg{w}$.
\item $\reveal{w}{\be}{\cid}$: This
  is a \emph{public reveal} where expression $\be$ is computed
  by client $\cid$ and broadcast to the federation, typically
  to communicate intermediate results for use in final output
  computations.
\item $\pubout{\cid}{\be}{\cid}$: This
  is an \emph{output} where expression $\be$ is computed
  by client $\cid$ and reported as its output. As a
  sanity condition we disallow commands
  $\pubout{\cid_1}{\be}{\cid_2}$ where $\cid_1\ne\cid_2$.
\end{itemize}
For example, in the following protocol, a client 1
subtracts a random sample $\flip{y}$ from $\mathbb{F}_p$ from their
secret value $\secret{x}$ and sends the result to client
2 as a message $\mesg{z}$:
$$
\eassign{\mesg{z}}{2}{(\secret{x} - \flip{y})}{1}
$$ Both messages $\mesg{w}$ and reveals $\rvl{w}$ can be referenced in
expressions once they've been defined.  This distinction between
messages and broadcast public reveal is consistent with previous
formulations, e.g., \cite{6266151}. To identify and distinguish
between collections of variables in protocols we introduce the
following notation.
\begin{definition}
We let $x$ range over \emph{variables} which are identifiers where
client ownership is specified- e.g.,
$\elab{\mesg{\mathit{foo}}}{\cid}$ is a message $\mathit{foo}$ that
was sent to $\cid$. We let $X$ range over sets of variables, and more
specifically, $S$ ranges over sets of secret variables
$\elab{\secret{w}}{\cid}$, $R$ ranges over sets of random variables
$\elab{\flip{w}}{\cid}$, $M$ ranges over sets of message variables
$\elab{\mesg{w}}{\cid}$, $P$ ranges over sets of public variables
$\rvl{w}$, and $O$ ranges over sets of output variables $\out{\cid}$.
Given a program $\prog$, we write $\iov(\prog)$ to denote the
particular set $S \cup M \cup P \cup O$ of variables in $\prog$ and
$\secrets(\prog)$ to denote $S$, and we write $\flips(\prog)$ to
denote the particular set $R$ of random samplings in $\prog$. We write
$\vars(\prog)$ to denote $\iov(\prog) \cup \flips(\prog)$. For any set
of variables $X$ and clients $I$, we write $X_I$ to denote the subset
of $X$ owned by any client $\cid \in I$, in particular we write $X_H$
and $X_C$ to denote the subsets belonging to honest and corrupt
parties, respectively.
\end{definition}

\subsection{Semantics}
\label{section-lang-semantics}

\emph{Memories} are fundamental to the semantics of $\fedcat$ and
provide random tape and secret inputs to protocols, and also record
message sends, public broadcast, and client outputs.
\begin{definition}
  Memories $\store$
are finite (partial) mappings from variables $x$ to values $v \in
\mathbb{F}_p$.  The \emph{domain} of a memory is written
$\dom(\store)$ and is the finite set of variables on which the memory
is defined.  We write $\store\{ x \mapsto v\}$ for
$x\not\in\dom(\store)$ to denote the memory $\store'$ such that
$\store'(x) = v$ and otherwise $\store'(y) = \store(y)$ for all $y \in
\dom(\store)$. We write $\store \subseteq \store'$ iff $\dom(\store)
\subseteq \dom(\store')$ and $\store(x) = \store'(x)$ for all $x \in
\dom(\store)$. Given any $\store$ and $\store'$ with
$\store(x) = \store'(x)$ for all $x \in \dom(\store) \cap \dom(\store')$,
we write $\store \uplus \store'$ to denote the memory
with domain $X = \dom(\store) \cup \dom(\store')$ such
that:
$$
\forall x \in X .
(\store \uplus \store')(x) =
\begin{cases} \store(x) \text{\ if\ } x\in\dom(\store) \\ \store'(x) \text{\ otherwise\ }\end{cases} 
$$
\end{definition}
In our subsequent presentation we will often want to consider arbitrary
memories that range over particular variables and to restrict
memories to particular subsets of their domain:
\begin{definition}
  Given a set of variables $X$ and memory $\store$, we write
  $\store_X$ to denote the memory with $\dom(\store_X) = X$ and
  $\store_X(x) = \store(x)$ for all $x \in X$. We define $\mems(X)$ as
  the set of all memories with domain $X$:
  $$
  \mems(X) \defeq \{ \store \mid \dom(\store) = X \}
  $$
\end{definition}
So for example, given a protocol $\prog$, the set of all random tapes for
$\prog$ is $\mems(\flips(\prog))$, and the memory $\store_{\secrets(\prog)}$
is $\store$ restricted to the secrets in $\prog$.

Given a variable-free expression $\be$, we write $\cod{\be}$ to denote
the standard interpretation of $\be$ in the arithmetic field
$\mathbb{F}_{p}$. With the introduction of variables to expressions,
we need to interpret variables with respect to a specific memory, and
all variables used in an expression must belong to a specified client.
Thus, we denote interpretation of expressions $\be$ computed on a
client $\cid$ as $\lcod{\store,\be}{\cid}$. This interpretation is
defined in Figure \ref{fig-minifed}. The small-step reduction relation
$\redx$ is then defined in Figure \ref{fig-minifed} to evaluate
commands. Reduction is a relation on \emph{configurations} $(\store,
\prog)$ where all three command forms- message send, broadcast, and
output- are implemented as updates to the memory $\store$. We write
$\redxs$ to denote the reflexive, transitive closure of\ $\redx$.

\subsection{Example: Passive Secure Addition}
\label{section-lang-example}

Shamir addition leverages homomorphic properties of addition in
arithmetic fields to implement secret addition. If a field value $v_1$
is uniformly random, then $v_1 \fminus v_2$ is an encryption of $v_2$
where $v_1$ is an information theoretically secure one-time-pad, which
is exploited for secret sharing, noting that $v_2$ can be
reconstructed by summing $v_1$ and $v_3 \defeq v_1 \fminus v_2$. 

In $\minifed$, to privately sum secret values $\secret{\cid}$, each
client $\cid$ in the federation $\{ 1, 2, 3 \}$ samples a value
$\locflip$ that can be used as a one-time pad with another random
sample $\flip{x}$ and $\secret{\cid}$. This yields two secret shares
communicated as messages to the other clients, while each client keeps
$\locflip$ as its own share.
$$
\begin{array}{lll}
  \elab{\mesg{s1}}{2} &:=& \elab{(\secret{1} \fminus \locflip \fminus \flip{x})}{1} \\ 
  \elab{\mesg{s1}}{3} &:=& \elab{\flip{x}}{1} \\ 
  \elab{\mesg{s2}}{1} &:=& \elab{(\secret{2} \fminus \locflip \fminus \flip{x})}{2} \\ 
  \elab{\mesg{s2}}{3} &:=& \elab{\flip{x}}{2} \\ 
  \elab{\mesg{s3}}{1} &:=& \elab{(\secret{3} \fminus \locflip \fminus \flip{x})}{3} \\ 
  \elab{\mesg{s3}}{2} &:=& \elab{\flip{x}}{3}
\end{array}
$$
This scheme guarantees that messages
are viewed as random noise by any observer 
besides $\cid$ \cite{barthe2019probabilistic}. Next, each client
publicly reveals the sum of all of its shares, including its local
share. This step does reveal information about secrets-- note in
particular that $\locflip$ is reused and is no longer a one-time-pad:
$$
\begin{array}{lll}
  \rvl{1} &:=& \elab{(\locflip \fplus \mesg{s2} \fplus \mesg{s3})}{1} \\ 
  \rvl{2} &:=& \elab{(\mesg{s1} \fplus \locflip \fplus \mesg{s3})}{2} \\
  \rvl{3} &:=& \elab{(\mesg{s1} \fplus \mesg{s2} \fplus \locflip)}{3} 
\end{array}
$$
Finally, each client outputs the sum of each sum of shares, yielding
the sum of secrets. The protocol is correct because the outputs are all the
true sum of secrets, and it is secure because no more information about the
secrets other than that revealed by their sum is exposed.
$$
\begin{array}{lll}
  \out{1} &:=& \elab{(\rvl{1} \fplus \rvl{2} + \rvl{3})}{1}\\
  \out{2} &:=& \elab{(\rvl{1} \fplus \rvl{2} + \rvl{3})}{2}\\
  \out{3} &:=& \elab{(\rvl{1} \fplus \rvl{2} + \rvl{3})}{3}
\end{array}
$$

\section{Security Model}
\label{section-pmf}
\label{section-model}

MPC protocols are intended to implement some \emph{ideal
functionality} $\idealf$ with per-client outputs. In the $\minifed$
setting, Given a protocol $\prog$ that implements $\idealf$, with
$\iov(\prog) = S \cup M \cup P \cup O$, the domain of $\idealf$
is $\mems(S)$ and its range is $\mems(O)$.  Real/ideal security in the MPC
setting means that, given $\store \in \mems(S)$, a secure protocol
$\prog$ does not reveal any more information about honest secrets
$\store_H$ to parties in $C$ beyond what is implicitly declassified by
$\idealf(\sigma)$. Security comes in \emph{passive} and
\emph{malicious} flavors, wherein the adversary either follows the
rules or not, respectively. Characterization of both real world
protocol execution and simulation is defined
probabilistically. Following previous work
\cite{barthe2019probabilistic} we use probability mass functions to
express joint dependencies between input and output variables, as a
metric of information leakage.

\subsection{Probability Mass Functions} 

We define discrete joint probability mass functions (pmfs) in a
standard manner but develop some notations that are useful for our
presentation. Firstly, whereas distributions typically map
realizations of random variables to values in $[0..1]$, for
convenience and following \cite{barthe2019probabilistic} we use
memories to represent realizations, so for example given a pmf $\pmf$
over variables $\{ \sx{x}{1}, \mx{y}{2} \}$ we write $\pmf(\{
\elab{\secret{x}}{1} \mapsto 0, \elab{\mesg{y}}{2} \mapsto 1 \})$ to
denote the (joint) probability that $\elab{\secret{x}}{1} = 0 \wedge
\elab{\mesg{y}}{2} = 1$. Recall from Section
\ref{section-lang-semantics} that $\uplus$ denotes the combination of
memories, so for example $\{ \elab{\secret{x}}{1} \mapsto 0\} \uplus
\{\elab{\mesg{y}}{2} \mapsto 1 \} = \{ \elab{\secret{x}}{1} \mapsto 0,
\elab{\mesg{y}}{2} \mapsto 1 \}$.
\begin{definition}
  A \emph{probability mass function} $\pmf$ is a function
  mapping memories in $\mems(X)$ for given variables $X$ to
  values in $\mathbb{R}$ such that:
  $$
  \sum_{\store \in \mems(X)} \pmf(\store) \  = \ 1
  $$
\end{definition}
Now, we can define a notion of marginal and conditional
distributions as follows, which are standard for discrete
probability mass functions. 
\begin{definition}
  Given $\pmf$ with $\dom(\pmf) = \mems(X_2)$, the \emph{marginal distribution}
  of variables $X_1 \subseteq X_2$ in $\pmf$ is denoted $\margd{\pmf}{X_1}$ and defined as follows:
  $$
  \forall \store \in \mems(X_1) \quad . \quad \margd{\pmf}{X_1}(\store) \defeq
  \sum_{\store' \in \mems(X_2-X_1)} \pmf(\store \uplus \store')
  $$
\end{definition}

\begin{definition}
  Given $\pmf$, the \emph{conditional distribution}
  of $X_1$ given $X_2$ where $X_1 \cup X_2 \subseteq \dom(\pmf)$ and $X_1 \cap X_2 = \varnothing$
  is denoted $\condd{\pmf}{X_1}{X_2}$ and defined as follows:
  $$
  \forall \store \in \mems(X_1 \cup X_2)\ .\ 
  \condd{\pmf}{X_1}{X_2}(\store) \defeq
  \begin{cases}
    0 \text{\ if\ } \margd{\pmf}{X_2}(\store_{X_2}) = 0\\
    \margd{\pmf}{X_1 \cup X_2}(\store) / \margd{\pmf}{X_2}(\store_{X_2})\ \text{\ o.w.}
  \end{cases}
  $$
\end{definition}
We also define some convenient syntactic sugarings. The first will allow us to
compare marginal distributions under different realization conditions
(as in, e.g., Definition \ref{definition-NIMO}), the others are standard.
\begin{definition}
  Given $\pmf$, for all $\store_1 \in \mems(X_1)$ and $\store_2 \in \mems(X_2)$ define:
  \begin{enumerate}
  \item $\condd{\pmf}{X_1}{\store_2}(\store_1) \defeq \condd{\pmf}{X_1}{X_2}(\store_1 \uplus \store_2)$
  \item $\pmf(\store_1)  \defeq \margd{\pmf}{X_1}(\store_1)$ 
  \item $\pmf(\store_1|\store_2) \defeq \condd{\pmf}{X_1}{X_2}(\store_1 \uplus \store_2)$
  \end{enumerate}
\end{definition}

\subsection{Basic Distribution of a Protocol}
Now we can define the probability distribution of a program $\prog$,
that we denote $\progtt(\prog)$. Since $\fedcat$ is deterministic the
results of any run are determined by the input values together with
the random tape. And since we constrain programs to not overwrite
views, we are assured that \emph{final} memories contain both a
complete record of all initial secrets as well as views resulting from
communicated information. 

Our semantics require that random tapes contain values for all program
values $\elab{\flip{w}}{\cid}$ sampled from a uniform distribution
over $\mathbb{F}_p$. Input memories also contain input secret values
and possibly other initial view elements as a result of
pre-processing, e.g., Beaver triples for efficient multiplication,
and/or MACed share distributions as in BDOZ/SPDZ
\cite{evans2018pragmatic,10.1007/978-3-030-68869-1_3}. We define
$\runs(\prog)$ as the set of final memories resulting from execution
of $\prog$ given any initial memory, and treat all elements of
$\runs(\prog)$ as equally likely.  This establishes the basic program
distribution that can be marginalized and conditioned to quantify
input/output information dependencies.
\begin{definition}
  \label{def-progtt}
  \label{def-progd}
  \label{definition-progd}
  Given $\prog$ with $\secrets(\prog) = S$ and $\flips(\prog) = R$ and
  pre-processing predicate $\preproc$ on memories, define:
  $$
  \begin{array}{c}
    \runs(\prog) \defeq \\
    \{ \store \mid \exists \store_1 \in \mems(R) . 
    \exists \store_2 . \preproc(\store_2) \wedge
    (\store_1 \uplus \store_2,\prog) \redxs (\store,\varnothing) \}
  \end{array}
  $$
  By default, $\preproc(\store) \iff \dom(\store) = S$, i.e.,
  the initial memory contains all input secrets in a uniform
  marginal distribution. Then the \emph{basic distribution of $\prog$}, written $\progtt(\prog)$, is
  defined such that for all $\store \in \mems(\iov(\prog) \cup R)$:
  $$
  \progtt(\prog)(\store) =  1 / |\runs(\prog)| \ \text{if}\ \store \in \runs(\prog), \text{otherwise}\ 0
  $$
  
\end{definition}

\subsection{Honest and Corrupt Views}

Information about honest secrets can be revealed to corrupt clients
through messages sent from honest to corrupt clients, and through
publicly broadcast information from honest clients. Dually,
corrupt clients can impact protocol integrity through the messages
sent from corrupt to honest clients, and through publicly broadcast information
from corrupt clients. We call the former \emph{corrupt views}, and
the latter \emph{honest views}. Generally we let $V$ range over sets
of views.
\begin{definition}[Corrupt and Honest Views]
  We let $V$ range over \emph{views} which are sets of messages
  and reveals. Given a program $\prog$ with $\iov(\prog) = S \cup M \cup P \cup O$,
  define $\views(\prog) \defeq M \cup P$, and define $\houtputs$ as
  the messages and reveals in $V = M \cup P$ sent from honest to corrupt
  parties, called \emph{corrupt views}:
  $$
  \begin{array}{lcl}
    \houtputs & \defeq
        & \{\ \rvl{w} \mid\ \reveal{w}{\be}{\cid} \in \prog \wedge \cid \in H \ \}\ \cup \\
      & & \{\ \elab{\mesg{w}}{\cid}\ \mid\  \eassign{\mesg{w}}{\cid}{\be}{\cid'} \in
           \prog \wedge \cid \in C \wedge \cid' \in H \ \} 
  \end{array}
  $$
  and similarly define $\cinputs$ as the subset of $V$ sent from corrupt to honest
  parties, called \emph{honest views}:
  $$
  \begin{array}{lcl}
    \cinputs &  \defeq
        & \{\ \rvl{w} \mid\ \reveal{w}{\be}{\cid} \in \prog \wedge \cid \in C \ \} \ \cup\\
      & & \{\ \elab{\mesg{w}}{\cid}\ \mid\  \eassign{\mesg{w}}{\cid}{\be}{\cid'} \in
              \prog \wedge \cid \in H \wedge \cid' \in C \ \}
  \end{array}
  $$
\end{definition}

\subsection{Passive Correctness and Security}

In the passive setting we assume that $H$ and $C$ follow the
rules of protocols and share messages as expected. A first
consideration is whether a given protocol is \emph{correct}
with respect to an ideal functionality. 
\begin{definition}[Passive Correctness]
  We say that a protocol $\prog$ is \emph{passive correct} for a functionality
  $\idealf$ iff for all $\store \in \mems(\secrets(\prog))$
  we have $\progtt(\prog)(\idealf(\store) \mid \store) = 1$.
\end{definition}

In the passive setting the simulator must construct a probabilistic
algorithm $\SIM$, aka a \emph{simulation}, that is parameterized by
corrupt inputs and the output of an ideal functionality, and that
returns a reconstruction of corrupt views that is probabilistically
indistinguishable from the corrupt views in the real world protocol
execution.
\begin{definition}
  Given $\store$, and $v$,we write $ \prob(\SIM(\store,v) = \store') $
  to denote the probability that $\SIM(\store,v)$ returns corrupt views
  $\store'$ as a result. We write $\dist(\SIM(\store,v))$ to
  denote the distribution of corrupt views reconstructed by the
  simulation, where for
  all $\store' \in \mems(V)$:
  $$
  \dist(\SIM(\store,v))(\store')\ \defeq\ \prob(\SIM(\store,v) = \store') 
  $$
\end{definition}
Then we can define passive security in the real/ideal
model as follows. 
\begin{definition}[Passive Security]
  Assume given a program $\prog$ that correctly implements an ideal
  functionality $\idealf$, with $\views(\prog) = V$.  Then $\prog$
  is \emph{passive secure in the simulator model} iff there exists
  a simulation $\SIM$ such that for all
  partitions of the federation into honest and corrupt sets $H$ and $C$
  and for all $\store \in \mems(\secrets(\prog))$:
  $$
  \dist(\SIM(\store_{C},\idealf(\store))) = \condd{\progtt(\prog)}{\houtputs}{\store}
  $$
\end{definition}

\subsection{Malicious Security}

In the malicious model we assume that corrupt clients are in
thrall to an adversary $\adversary$ who does not necessarily follow
the rules of the protocol.  We model this by positing a $\arewrite$
function which is given a corrupt memory $\store_C$ and expression
$\be$, and returns a rewritten expression that can be interpreted to
yield a corrupt input. We define the evaluation relation that
incorporates the adversary in Figure \ref{fig-adversary}.

\adversaryfig

A key technical distinction of the malicious setting is that it
typically incorporates ``abort''. Honest parties implement strategies
to detect rule-breaking-- aka \emph{cheating}-- by using, e.g.,
message authentication codes with semi-homomorphic properties as in
BDOZ/SPDZ \cite{10.1007/978-3-030-68869-1_3}. If cheating is detected,
the protocol is aborted. To model this, we extend $\minifed$ with an
\ttt{assert} command and extend the range of memories with
$\bot$. Note that the adversary is free to ignore their own
assertions.
\begin{definition}
  We add assertions of the form $\elab{\assert{\phi(\be)}}{\cid}$ to the command
  syntax of $\minifed$, where $\phi$ is a decidable predicate on
  $\mathbb{F}_p$ and with operational semantics given in Figure
  \ref{fig-adversary}. We also extend the range of memories $\store$
  to $\mathbb{F}_p \cup \{ \bot \}$.
\end{definition}

It is necessary to add $\bot$ to the range of memories since
the possibility of abort needs to be reflected in adversarial
runs of a protocol. We can define $\aruns(\prog)$
as the ``prefix'' memories that result from possibly-aborting
protocols, but we also need to ``pad out'' the memories
of partial runs with $\bot$, as we define in $\botruns(\prog)$,
to properly reflect the contents of views and outputs even in case of abort. 
\begin{definition}
  \label{def-aprogd}
  \label{def-aprogtt}
  \label{definition-aprogd}
  Given program $\prog$ with $\iov(\prog) = S \cup V \cup O$ and $\flips(\prog) = R$, and
  any assumed pre-processing predicate $\preproc$ on memories, define:
  $$
  \begin{array}{c}
    \aruns(\prog) \defeq \\
    \{ \store \mid \exists \store_1 \in \mems(R) . 
    \exists \store_2 . \preproc(\store_2) \wedge
    (\store_1 \uplus \store_2,\prog) \aredxs (\store,\varnothing) \}
  \end{array}
  $$
  where by default, $\preproc(\store) \iff \dom(\store) = S$, and also define the following
  which pads out undefined views and outputs with $\bot$:
  $$
  \begin{array}{l}
    \botruns(\prog) \defeq \\
    \qquad \{ \store\{ x_1 \mapsto \bot, \ldots, x_n \mapsto \bot \} \mid \\
    \qquad \phantom{\{} \store \in \aruns(\prog) \wedge \{ x_1,\ldots,x_n\} = (V \cup O) - \dom(\store) \}
  \end{array}
  $$
  Then the \emph{$\adversary$ distribution of $\prog$}, written $\progtt(\prog,\adversary)$, is
  defined such that for all $\store \in \mems(\iov(\prog) \cup R)$:
  $$
  \progtt(\prog,\adversary)(\store) =  1 / |\botruns(\prog)| \ \text{if}\ \store \in \botruns(\prog), \text{otherwise}\ 0
  $$
\end{definition}

Given this preamble, we can define malicious simulation and malicious security
in a standard manner \cite{evans2018pragmatic}, as follows.
\begin{definition}[Malicious Simulation]
  Given a protocol $\prog$ with $\iov(\prog) = S \cup V \cup O$, honest and corrupt 
  clients $H$ and $C$, adversary $\adversary$, and honest inputs
  $\store \in \mems(S_H)$, the \emph{malicious simulation}  $\SIM(\store)$ has three phases:
  \begin{enumerate}
  \item In the first phase $\SIM_1$, $\adversary$ gives the
    simulator some $\store' \in \mems(S_C)$, and the simulator consults an
    oracle to compute $\idealf(\store \uplus \store') \in \mems(O)$.
  \item In the second phase $\SIM_2$, the simulator is given the corrupt
    outputs $\idealf(\store \uplus \store')_C$, which are again given to
    $\adversary$, who decides either to abort or not. If so, then the
    simulator is given $\sigma_{\mathit{out}} \defeq \{ \out{\cid} \mapsto \bot \mid \cid \in H \}$
    and arbitrary internal state $\varsigma$.
    Otherwise the simulator is given $\sigma_{\mathit{out}} \defeq \idealf(\store \uplus \store')_H$
    and $\varsigma$.
  \item In the third phase $\SIM_3$, given $\store_{\mathit{out}}$ and $\varsigma$, the simulator
    finally outputs
    $\store_{\mathit{out}} \uplus \store_{\mathit{views}}$ for some
    calculated $\store_{\mathit{views}} \in \mems(\houtputs)$.
  \end{enumerate}
\end{definition}

\begin{definition}[Malicious Security]
  We write $\dist(\SIM(\store))$ to
  denote the distribution of honest outputs and corrupt views reconstructed by the
  malicious simulation, where for
  all $\store'$:
  $$
  \dist(\SIM(\store))(\store')\ \defeq\ \prob(\SIM(\store) = \store') 
  $$
  Then a protocol $\prog$ with $\iov(\prog) = S \cup V \cup O$ is \emph{malicious
  secure} iff for all $H$, $C$, $\adversary$, and $\store \in \mems(S_H)$:
  $$
  \dist(\SIM(\store)) = \condd{\progtt(\prog,\adversary)}{\houtputs \cup O_H}{\store}
  $$  
\end{definition}

\section{Security Hyperproperties}
\label{section-hyper}

In this Section we formulate probabilistic versions of well-studied
hyperproperties of confidentiality and integrity, including
noninterference, gradual release, declassification, and robust
declassification.  We follow nomenclature developed in previous work
on characterizing declassification policies in deterministic settings
\cite{sabelfeld2009declassification}, but adapt them to our
probabilistic one.


\subsection{Conditional Noninterference}

Since MPC protocols release some information about secrets through
outputs of $\idealf$, they do not enjoy strict noninterference.  As
discussed in Section \ref{section-lang}, public reveals and protocol
outputs are fundamentally forms of declassification.  But consistent
with other work \cite{8429300}, we can formulate a version of
probabilistic noninterference conditioned on output that is sound
for passive security. 
\begin{definition}[Noninterference modulo output]
  \label{definition-NIMO}
  We say that a program $\prog$ with $\iov(\prog) = S \cup V \cup O$
  satisfies \emph{noninterference modulo output}
  iff for all $H$ and $C$ and $\store_1 \in \mems(S_C \cup O)$ and $\store_2 \in \mems(\houtputs)$
  we have:
  $$
  \condd{\progtt(\prog)}{S_H}{\store_1} = \condd{\progtt(\prog)}{S_H}{\store_1 \uplus \store_2}
 $$
\end{definition}
This conditional noninterference property implies that
corrupt views give the adversary no better chance of guessing honest
secrets than just the output and corrupt inputs do. So the simulator
can just arbitrarily pick any honest secrets that could have produced
the given outputs and run the protocol in simulation to reconstruct
real world corrupt views. This requires that the simulator can
tractably ``pre-image'' a given output of a functionality $\idealf$,
to determine the inputs that could have produced it. This equivalence
class is called a \emph{kernel} in recent work \cite{10.1145/3571740}.
\begin{definition}
  Given a functionality $\idealf$ and outputs $\store_{\mathit{out}}$, their 
  \emph{kernel}, denoted $\kernel{\idealf}{\store_{\mathit{out}}}$ is
  $
  \{ \store\ |\ \idealf(\store) = \store_{\mathit{out}} \}
  $.
  We say that $\idealf$ is \emph{pre-imageable} iff $\kernel{\idealf}{\store_{\mathit{out}}}$ for all
  $\store_{\mathit{out}}$ can be computed tractably.
\end{definition}
A soundness result for passive security can then be given as follows.
It is essentially the same as ``perfect passive NI security'' explored
in previous work \cite{8429300}.  
\begin{theorem}
  \label{theorem-nimo}
  Assume given pre-imageable $\idealf$ and a protocol $\prog$ that
  correctly implements $\idealf$.  If $\prog$ satisfies noninterference modulo output
  then $\prog$ is passive secure.
\end{theorem}

\subsection{Gradual Release}

Probabilistic noninterference is related to perfect secrecy and is
preserved by components of cryptographic protocols generally. It can
be expressed using probabilistic independence, aka separation,
\cite{darais2019language,barthe2019probabilistic}, and we adopt the
following notation to express independence:
\begin{definition}
  We write $\sep{\pmf}{X_1}{X_2}$ iff for all
    $\store \in \mems(X_1 \cup X_2)$ we have
  $\margd{\pmf}{X_1 \cup X_2}(\store) =
  \pmf(\store_{X_1}) * \pmf(\store_{X_2})$
\end{definition}

In practice, MPC protocols typically satisfy a \emph{gradual
release} property \cite{sabelfeld2009declassification}, where messages
exchanged remain probabilistically separable from secrets, with only
declassification events (reveals and outputs) releasing information
about honest secrets.  A key difference is that while these
declassification events essentially define the policy in gradual
release, the ideal functionality sets the release policy for MPC
passive security, so its necessary to show that declassification
events respect these bounds.
\begin{definition}
  Given $H,C$, a protocol $\prog$ with $\iov(\prog) = S \cup M \cup P \cup O$
  satisfies \emph{gradual release} iff
  $\sep{\progtt(\prog)}{M_C}{S_H}$.
\end{definition}

\subsection{Integrity and Robust Declassification}


\emph{Integrity} is an important hyperproperty in security models that admit
malicious adversaries. Consistent with formulations in deterministic settings,
we have already defined protocol confidentiality as the preservation of low equivalence
(of secrets and views), and now we define protocol integrity as the preservation
of high equivalence (of secrets and views). Intuitively, this property says
that any adversarial strategy either ``mimics'' a passive strategy with some
choice of inputs or causes an abort.
\begin{definition}[Integrity]
  \label{def-integrity}
  We say that a protocol $\prog$ with $\iov(\prog) = S \cup V \cup O$ has
  \emph{integrity} iff for all $H$, $C$, and $\adversary$,
  if $\store \in \aruns(\prog)$ 
  then there exists $\store' \in \mems(S)$ with $\store_{S_H} = \store'_{S_H} $ and:
    $$
    \condd{\progtt(\prog,\adversary)}{X}{\store_{S_H \cup \cinputs}} =
    \condd{\progtt(\prog)}{X}{\store'}
    $$ 
  where $X \defeq (\houtputs \cup O_H) \cap \dom(\store)$. 
\end{definition}
A first important observation is that integrity preserves protocol correctness
for honest outputs, except for the possibility of abort. 
\begin{lemma}
  \label{lemma-malicious-correct}
  If a protocol $\prog$ with $\iov(\prog) = S \cup V \cup O$ is passive correct for
  $\idealf$ and
  has integrity, then for all $H$, $C$, $\adversary$, $\store_1 \in \mems(S_H)$,
  $\ox{\cid} \in (O_H)$, and $\mv \in \mathbb{F}_p$, if:
  $$
  \progtt(\prog,\adversary)(\{ \ox{\cid} \mapsto \mv \} \mid \store_1) > 0
  $$
  then exists $\store_2 \in \mems(S_C)$ such that:
  $$
  \idealf(\store_1 \uplus \store_2)(\ox{\cid}) = \mv
  $$
\end{lemma}
The following result establishes that integrity implies malicious
security for protocols that are passive secure (which also subsumes
correctness). 
\begin{theorem}
  \label{theorem-integrity}
  If a protocol is passive secure and has integrity, then it
  is malicious secure.
\end{theorem}

\begin{proof}
  Let $\prog$ be some protocol with passive security and integrity
  where $\iov(\prog) = S \cup V \cup O$, and let $\adversary$ be some
  adversary. Suppose $\store \in \aruns(\prog)$.
  As integrity requires, and as Lemma \ref{lemma-malicious-correct}
  demonstrates with respect to outputs, the most the adversary can do
  in the presence of integrity is to elicit the same responses from
  the honest parties-- via the strategy $\store_{\cinputs}$-- as
  are elicited from some passive run of the protocol using
  some $\store' \in \mems(S)$ where $\store'_H = \store_{S_H}$,
  and perhaps to force an abort after some number of message
  exchanges.

  Therefore, in simulation, $\adversary$ can provide the simulator
  with some $\store'_C$ in $\SIM_1$ which its strategy is ``impersonating'',
  allowing $\idealf(\store')$ to
  be communicated to $\adversary$ in $\SIM_2$ who can then
  decide whether or not to abort. In the case of abort, the
  subset of $\houtputs$ to be defined can be communicated to
  $\SIM_3$, along with $\idealf(\store')$, via $\Sigma$. 
  In $\SIM_3$, the simulator can then run $\prog$ in simulation
  with inputs $\store'_C$ and arbitrary $\store'' \in \mems(S_H)$
  such that $\store'_C \uplus \store'' \in \kernel{\idealf}{\idealf(\store')}$.
  The assumption of passive security of $\prog$ implies the result.
\end{proof}

The hyperproperty of robust declassification \cite{930133} similarly
combines a confidentiality property with integrity to establish that
malicious actors cannot declassify more information than is intended
by policy. But in this prior work, this policy is established
by the declassifications themselves, as in gradual release.
Thus, we can define a robust declassification property as follows. 
\begin{definition}[Robust Declassification]
  A protocol satisfies \emph{robust declassification} iff it has integrity and
  satisfies gradual release. 
\end{definition}
However, it is important to note that gradual release, and
hence robust declassification, are not sufficient to establish
passive or malicious simulator security, where the declassification
policy is established by the ideal functionality $\idealf$. 
\begin{theorem}
  Robust declassification does not imply malicious security, but
  passive security with robust declassification implies malicious security.
\end{theorem}

\section{Automated Verification in $\mathbb{F}_2$}
\label{section-bruteforce}

In the binary field $\mathbb{F}_2$, a brute force strategy for computing
$\progtt(\prog)$ for any $\prog$ is to directly compute
$\runs(\prog)$. By querying $\progtt(\prog)$ we can verify any of the
hyperproperties discussed previously, or other properties such as
perfect secrecy \cite{barthe2019probabilistic}.  A basic method to do this is
to calculate a truth table for the given protocol. Since
$\runs(\prog)$ is exponential in
the size of $\secrets(\prog) \cup \flips(\prog)$, this strategy is
feasible only for smaller protocols. However, to support some
scaling we can efficiently
convert protocols in the passive setting (without $\ttt{assert}$) to
stratified Datalog programs, and then extract $\runs(\prog)$ by
calculating Least Herbrand models by parallelization and/or other HPC
acceleration techniques for logic programs \cite{aspis2018linear}.
The rewriting we describe here is to Datalog with negation, with a
negation-as-failure model, though we can also use techniques in
\cite{sakama2017linear} to eliminate negation from resulting programs.

\subsection{Computing a Truth Table}

Letting $\stores$ denote sets of memories, in Figure \ref{fig-solve}
we define the algorithm $\solve{\stores}{\be}{\cid}$ which filters a
given $\stores$ to obtain the subset whose elements satisfy $\be$. In
this Figure and elsewhere we use logical connectives as field
operations ($\eand$ and $\exor$ for $*$ and $+$ respectively) in
$\mathbb{F}_2$ and add $\enot$ and $\eor$ as trivial but convenient
extensions. The correctness of this operation is characterized as
follows.
\begin{lemma}
  \label{lemma-solves}
  $(\solve{\stores}{\be}{\cid}) = \{ \store \in \stores \ \mid\ \lcod{\store,\be}{\cid} = 1 \}$.
\end{lemma}

\solvefig

We immediately note that $\runs(\prog)$ can be obtained by a left-folding
of $\solvealg$ across $\prog$. 
\begin{lemma}
  \label{lemma-cruns}
  Given $\prog$ where $\iov(\prog) = S \cup V \cup O$ and $\flips(\prog) = R$. Define:
  \begin{eqnarray*}
    {tt}\ \ \stores\ (\xassign{x}{\be}{\cid}) &\defeq& \begin{array}{l}
      \mathrm{let}\ \stores' = \solve{\stores}{\be}{\cid}\ \mathrm{in}\\
      \ \ \{\extend{\store}{x}{1} \mid \store \in \stores' \}\ \cup\\
      \ \ \{\extend{\store}{x}{0} \mid \store \in \stores - \stores' \}\end{array}
  \end{eqnarray*}
  Then assuming default preprocessing, $\mathit{foldl}\ {tt}\ \mems(S \cup R)\ \prog = \runs(\prog)$.
\end{lemma}
However, this method does not take advantage of parallelization,
in that elements of $\runs(\prog)$ can be calculated independently.

\subsection{Conversion to Stratified Datalog}

We define the syntax of Datalog as follows. As per standard
nomenclature, \emph{atoms} are $\minifed$ variables $x$,
\emph{literals} are atoms or negated atoms, and clause bodies are
conjunctions of literals.  A \emph{fact} is a clause with no body. A
\emph{Datalog program} is a list of clauses.
$$
\begin{array}{rclr}
  \mathit{body} &::=&  x \mid \neg x \mid x \wedge \mathit{body} \mid \neg x \wedge \mathit{body} \\
  \mathit{clause} &::=& x \gets \mathit{body} \mid x \gets
\end{array}
$$
When translating protocols, we need to extract the variables
that occur in expressions $\be$ computed by a client $\cid$,
written $\vars\ \be\ \cid$ where:
\begin{mathpar}
  \vars\ \secret{w} \cid \defeq \{ \elab{\secret{w}}{\cid} \}
  
  \vars\ \mesg{w} \cid \defeq \{ \elab{\mesg{w}}{\cid} \}

  \vars\ \rvl{w} \cid \defeq \{ \elab{\rvl{w}}{\cid} \}

  \vars\ \flip{w} \cid \defeq \{ \elab{\flip{w}}{\cid} \}

  \vars\ (\be_1 \exor \be_2)\ \cid \defeq (\vars\ \be_1\ \cid) \cup (\vars\ \be_2\ \cid)
\end{mathpar}
... and so on. Then, to convert a protocol $\prog$ to a Datalog
we first define the function $\bodies$ that applies $\solvealg$ locally to each command
in $\prog$, obtaining the subset of memories that result
in a variable assignment of $1$.  
\begin{definition} Define:
$$
\bodies(\xassign{x}{\be}{\cid}) \defeq (x, (\solve{(\mems(\vars\ \be\ \cid))}{\be}{\cid}))
$$
\end{definition}
The mapping of $\bodies$ across a program
$\prog$-- i.e., $(\mathit{map}\ \bodies\ \prog)$--  essentially defines the
logic program, modulo some syntactic conversion. We can
accomplish the latter as follows, where $\datalog(\prog)$ defines the
full conversion.
\begin{definition} We define the conversion from memories to
  literals and clause bodies as follows:
\begin{mathpar}
  \logit{x \mapsto 1} \defeq x

  \logit{x \mapsto 0} \defeq \neg x

  \logit{\{ x_1 \mapsto \beta_1, \ldots, x_n \mapsto \beta_n\}} \defeq
  \logit{x_1 \mapsto \beta_1} \wedge \cdots \wedge \logit{x_n \mapsto \beta_n}
\end{mathpar}
Given pairs $(x,\stores)$ in the range of $\bodies$, we define the conversion
to clauses as  $\mathit{clauses}(x,\{ \store_1,...,\store_n \}) \defeq x \gets \logit{\store_1} \vee \cdots \vee x \gets \logit{\store_n}$.
The $\minifed$-to-Datalog conversion is then defined as:
$$
\datalog(\prog) \defeq  \mathit{map}\ \mathit{clauses}\ (\mathit{map}\ \bodies\ \prog)
$$
\end{definition}

In addition to converting view definitions to logic clauses, we also need to convert
secrets and random tapes. Since we assume given values for these in an arbitrary run of
the program, we can capture these as a particular fact base.
\begin{definition}
  Given $\store$, let $\{x_1,\ldots,x_n \} =
  \{ x \in \dom(\store) \mid \store(x) = 1 \}$.
  Then define $\mathit{facts}(\store) \defeq x_1 \gets, \ldots, x_n \gets$.
\end{definition}
The following result ties these pieces together and establishes
correctness of this approach.
\begin{lemma}
  For all $\prog$, 
  $\datalog(\prog)$ is a \emph{normal}, \emph{stratified}
  program \cite{aspis2018linear}, and $\store$ is the unique Least Herbrand
  Model of: $$(\mathit{facts}(\store_{\secrets(\prog) \cup \flips(\prog)}),\datalog(\prog))$$
  iff $\store \in \runs(\prog)$.
\end{lemma}
Finally, to compute $\runs(\prog)$, and thus $\progtt(\prog)$, we compute
the Least Herbrand Model $\store$ of $(\mathit{facts}(\store'),\datalog(\prog))$
for all $\store' \in \mems(\secrets(\prog) \cup \flips(\prog))$, observing
that model computation for distinct fact bases can be done in parallel. 

\paragraph{Verifying Security Properties} We can
query correct representations of $\runs(\prog)$ by using
implementations of conditioning and marginalization to automatically
verify the passive model hyperproperties described in Section
\ref{section-hyper}-- in particular correctness and noninterference
modulo output, which imply passive security.  We have used this method
to verify security in $\mathbb{F}_2$ of Shamir addition as defined in Section
\ref{section-lang-example}, and single-gate $\eand$ and $\exor$ circuits with GMW
Beaver triples. We have also
verified properties compositional properties of gates themselves that establish
circuit invariants that imply passive security in larger circuits.
We discuss these examples and proof methods in Section
\ref{section-example-gmw}.

\section{The $\metaprot$ Metalanguage}
\label{section-metalang}

Practical MPC computations protocols are
typically composed of compositional units. Examples include GMW circuits
and Yao's Garbled Circuits (YGC), that are composed of so-called
garbled gates. Languages such as Fairplay \cite{269581} provide gates as
units of abstraction that are ``wired'' together by the programmer to
generate a complete circuit.

The $\fedprot$ language is low-level and does not include abstractions
for defining composable elements. So in this Section we introduce the
$\metaprot$ language that includes structured data and function
definitions for defining composable protocol elements at a higher
level of abstraction.  The $\metaprot$ language is a
\emph{metalanguage}, where $\fedprot$ protocols are the residuum of
computation. In addition to these declarative benefits of $\metaprot$,
component definitions support compositional verification of larger
protocols as we will discuss with examples in Sections
\ref{section-example-gmw} and \ref{section-example-bdoz}.

\metaprotfig

\subsection{Syntax}

The syntax of $\metaprot$ is defined in Figure
\ref{fig-metaprot}.  It includes a syntax of function
definitions and records, and values $\mv$ include client ids $\cid$, identifier
strings $w$, expressions $\be$ in field $\mathbb{F}_p$, and the unit value $()$.
Expression forms allow dynamic construction of field expressions and $\minifed$ commands.
The construction of a command has the side-effect of adding to the residual
$\minifed$ protocol. Formally, we consider a complete metaprogram to include both a
codebase and a ``main'' program that uses the codebase. 
\begin{definition}
A \emph{codebase} $\codebase$ is a list of function 
declarations. We write $ \codebase(f) = y_1,\ldots,y_n,\ e$
iff $f(y_1,\ldots,y_n) \{ e \} \in \codebase$.
A \emph{metaprogram}, aka \emph{metaprotocol}  is a pair of a 
codebase and expression $\codebase, e$. We may omit
$\codebase$ if it is clear from context.  
\end{definition}

When we consider larger examples, our
focus will be on developing a codebase that can be used to define
arbitrary circuits, i.e., complete and concrete protocols. Since
strings and identifiers can be constructed manually, and expressions
can occur inside assignments and field expression forms, function
definitions can generalize over $\fedprot$-level patterns to obtain
composable program units.
\subsection{Semantics}

We define a small-step evaluation aka reduction relation $\redx$ in
Figure \ref{fig-metaprot}.  We write $\redxs$ to denote the
reflexive, transitive closure of $\redx$. Reduction is defined on
\emph{configurations} which are pairs of the form $\config{\prog}{e}$,
where $\prog$ is the $\minifed$ program accumulated during evaluation.
In this definition we write $e[\mv/y]$ to denote the substitution of $\mv$
for free occurrences of $y$ in $e$. The rules are mostly standard,
except when a concrete $\minifed$ assignment is encountered it is added
to the end of $\prog$.

The rules rely on a definition of \emph{evaluation contexts} $E$
allowing computation within a larger program context, where $E[e]$
denotes an expression with $e$ in the hole $[]$ of $E$. The syntax
of $E$ imposes a left-to-right order of evaluation of subexpressions
for all forms.

\section{2-Party GMW and Passive Security Proof Tactics}
\label{section-metalang-gmw}
\label{section-example-gmw}

\begin{fpfig}[t]{2-party GMW circuit library.}{fig-gmw}
{\footnotesize
  \begin{verbatimtab}
    encodegmw(in, i1, i2) {
      m[in]@i2 := (s[in] xor r[in])@i2;
      m[in]@i1 := r[in]@i2;
      m[in]
    }
    
    andtablegmw(b1, b2, r) {
      let r11 = r xor (b1 xor true) and (b2 xor true) in
      let r10 = r xor (b1 xor true) and (b2 xor false) in
      let r01 = r xor (b1 xor false) and (b2 xor true) in
      let r00 = r xor (bl xor false) and (b2 xor false) in
      { row1 = r11; row2 = r10; row3 = r01; row4 = r00 }
    }
    
    andgmw(z, x, y) {
      let r = r[z] in
      let table = andtablegmw(x,y,r) in
      m[z]@2 := OT4(x,y,table,2,1);
      m[z]@1 := r@1;
      m[z]
    }
    
    xorgmw(z, x, y)
    { m[z]@1 := (x xor y)@1; m[z]@2 := (x xor y)@2; m[z] }
    
    decodegmw(z) {
      p["1"] := z@1; p["2"] := z@2;
      out@1 := (p["1"] xor p["2"])@1;
      out@2 := (p["1"] xor p["2"])@2
    }
  \end{verbatimtab}
}
\end{fpfig}

As an extended example of our language and security model, and how the
automated techniques in Section \ref{section-bruteforce} can serve
as tactics integrated with PSL/Lilac-style proofs, we consider GMW
circuits.  The GMW protocol is a garbled binary circuit protocol.  We
will assume the 2-party version, though it generalizes to $n$
parties \cite{goldreich2019play}. GMW uses a common technique in MPC, which is to
represent values $v$ as distributed shares $v_1$ and $v_2$ with $v =
v_1 \fplus v_2$. This trick maintains secrecy of $v$ from both
parties, and in GMW it is used to maintain the intermediate values of
internal gate outputs in circuits. In related literature the notation
$\macgv{x}$ is used to represent the ``true'' value of $x$ and $[x]$
is often used to represent the share of given party.

To capture this convention, which is used in many other protocols, we
introduce a new naming convention for ``global view'' elements
$\macgv{\mesg{w}}$, which denote the summed value of
$\elab{\mesg{w}}{1}$ and $\elab{\mesg{w}}{2}$ in a protocol
run. This concept integrates program distributions in the
usual manner, as the probability of the outcome of summation
of two variables in the distribution.
\begin{definition}
  For all $\mesg{w}$ define:
  $$\pmf(\macgv{\mesg{w}} = v) \defeq \sum_{\sigma \in A} \pmf(\sigma)$$
  and also for all $\store' \in \mems(X)$ define:
  $$\condd{\pmf}{X}{\macgv{\mesg{w}} = v}(\sigma) \defeq  \sum_{\sigma \in A} \condd{\pmf}{X}{\sigma}(\sigma')$$
  where $A$ is:
  $$\{ \store \in \mems(\{ \elab{\mesg{w}}{1},\elab{\mesg{w}}{2} \} ) \mid
      \fcod{\store(\elab{\mesg{w}}{1}) + \store(\elab{\mesg{w}}{2})} = v \}$$
\end{definition}

For full details of the GMW protocol the reader is referred to
\cite{evans2018pragmatic}. Our implementation library is shown in
Figure \ref{fig-gmw}, and includes encoding functions, where
input secrets are split into shares, $\eand$ and $\exor$ gate
functions, and a decoding function. Note that $\exor$ requires
no interaction between parties, while conjunction necessitates
1-of-4 oblivious transfer. The gate computation is
done entirely in secret, and the decoding function
is where the declassification occurs-- both parties reveal
their shares of the final gate output $\macgv{z}$.

For example, the following program uses our GMW library to define
a circuit with a single \eand\ gate and input secrets $\ttt{s1}$ and
$\ttt{s2}$ from client's 1 and 2 respectively:
\begin{verbatimtab}
         let s1 = encodegmw("s1",2,1) in
         let s2 = encodegmw("s2",1,2) in
         decodegmw(andgmw("z",s1,s2))
\end{verbatimtab}
By convention we will assume that all gates are assigned unique output
identifiers, and that all programs are in the form
of a sequence of let-bindings followed by a call to $\decodegmw$
wrapping a circuit.

\paragraph{Oblivious Transfer} A passive secure oblivious transfer (OT) protocol
based on previous work \cite{barthe2019probabilistic} can be defined in $\metaprot$,
however this protocol assumes some shared randomness. Alternatively,
a simple passive secure OT can be defined with the addition of
public key cryptography as a primitive. But given the diversity
of approaches to OT, we instead assume that OT is abstract with
respect to its implementation, where calls to OT in $\mathbb{F}_2$
are of the following form-- given a \emph{choice bit}
$\be_1$ provided by a receiver $\cid$, the sender
sends either $\be_2$ or $\be_3$.
$$
\OT{\elab{\be_1}{\cid}}{\be_2}{\be_3}
$$
Critically, the sender learns nothing about $\be_1$ and the
receiver learns nothing about the unselected value, so we interpret
these calls in our implementation in $\mathbb{F}_2$ as follows.
$$
\begin{array}{l}
\solve{\stores}{\OT{\elab{\be_1}{\cid_1}}{\be_2}{\be_3}}{\cid_2} = \\
\qquad ((\solve{\stores}{\be_1}{\cid_1}) \cap 
(\solve{\stores}{\be_2}{\cid_2})) \cup \\
\qquad ((\stores - (\solve{\stores}{\be_1}{\cid_1})) \cap
(\solve{\stores}{\be_3}{\cid_2})
\end{array}
$$

\subsection{Correctness Proof with Verification Tactics}

As discussed above and in related work \cite{8429300}, probabilistic
separation conditional on certain variables-- e.g., secret inputs or
public outputs-- is a key mechanism for reasoning about MPC protocol
security. Following \cite{barthe2019probabilistic}, we define a
conditional separation relation $\condsep{\pmf}{X_1}{X_2}{X_3}$ to
mean that $X_2$ and $X_3$ are independent in $\pmf$ conditionally on
any assignment of values to $X_1$-- i.e., conditionally on any $\store
\in \mems(X_1)$. Another key concept needed especially for reasoning
about circuits is conditional determinism. For example, if $\macgv{z}$
is an output of an internal gate, it will definitely be computed using
random variables, however, it \emph{should} be deterministic under any
set of input secrets $S$, since we assume that $\idealf$ is
deterministic. Conditional uniformity is also important, since the
gradual release property of many protocols means that messages appear
in a uniform distribution to the adversary.
\begin{definition}[Conditioning Properties] 
  Given $\pmf$ and $X_1,X_2,X_3 \subseteq \dom(\pmf)$, 
  we write:
  \begin{itemize}
  \item $\condsep{\pmf}{X_1}{X_2}{X_3}$ iff for all
    $\store' \in \mems(X_1)$ and $\store \in \mems(X_2 \cup X_3)$ we have
    $\pmf(\store|\store') = \pmf(\store_{X_2}|\store') *  \pmf(\store_{X_3}|\store')$.
  \item $\conddetx{\pmf}{X_1}{X_2}$ iff for all
    $\store' \in \mems(X_1)$ there exists 
    $\store \in \mems(X_2)$ such that $\pmf(\store|\store') = 1$.
  \item $\condunix{\pmf}{X_1}{X_2}$ iff for all
    $\store' \in \mems(X_1)$ and
    $\store \in \mems(X_2)$ we have
    $\pmf(\store|\store') = 1/p^{|X_2|}$.
  \end{itemize}
\end{definition}

Given these definitions, we can formulate an invariant for circuit
computation with respect to internal gates as follows. It says that
the output of any gate is deterministic given inputs $S$, and
conditionally on $S$ corrupt views are always in a uniform random
distribution (pure noise), while output $\macgv{\mesg{z}}$ remains
separable from corrupt views and both shares of
$\macgv{\mesg{z}}$. This last condition (incidentally missing from PSL
due to the reliance on more recent innovations in conditioning logic
\cite{li2023lilac}) is critical since those shares will in fact be
revealed if $\macgv{\mesg{z}}$ is decoded as the circuit output.
\begin{lemma}[GMW Invariant]
  \label{lemma-gmwinvariant}
  Given:
  $$ (\varnothing,e) \redxs (\prog,\decodegmw(E[\mesg{z}])) $$
  Then all of the following conditions hold for all $H$ and $C$ where $\iov(\prog) = S \cup M$:
  \begin{enumerate}
  \item $\conddetx{\progtt(\prog)}{S}{\{\macgv{\mesg{z}}\}}$
  \item $\condunix{\progtt(\prog)}{S}{M_C}$
  \item $\condsep{\progtt(\prog)}{S}{\{\macgv{\mesg{z}}\}}{\{ \elab{\mesg{z}}{1}, \elab{\mesg{z}}{2} \})}$
  \end{enumerate}
\end{lemma}
To prove this, we can formulate and automatically prove local,
gate-level versions of the invariant. This serves as a proof tactic
that simplifies the proof of the GMW invariant. 
\begin{lemma}[And Gate Tactic]
  \label{lemma-gmwtactic}
  Given:
  $$
  \begin{array}{c}
  (\varnothing,\andgmw(z,\mesg{x},\mesg{y}) \redxs 
  (\prog,\mesg{z})
  \end{array}
  $$
  Then all of the following conditions hold for both $\cid \in \{ 1,2 \}$ where $\iov(\prog) = M$:
  \begin{enumerate}
  \item
    $\conddetx{\progtt(\prog)}{\{ \macgv{\mesg{x}},\macgv{\mesg{y}} \}}{\{ \macgv{\mesg{z}} \}}$
  \item $\condunix{\progtt(\prog)}{\{ \macgv{\mesg{x}},\macgv{\mesg{y}} \}}{\{ \elab{\mesg{z}}{\cid} \}}$
  \item $\condsep
    {\progtt(\prog)}
    {\{ \macgv{\mesg{x}},\macgv{\mesg{y}} \}}
    {\{ \macgv{\mesg{z}} \}}
    {\{ \elab{\mesg{z}}{1},\elab{\mesg{z}}{2} \}}$
  \end{enumerate}
\end{lemma}
\begin{proof}
Verified automatically using techniques described in Section \ref{section-bruteforce}.  
\end{proof}

To properly integrate the local reasoning of Lemma \ref{lemma-gmwtactic} with
the global reasoning of Lemma \ref{lemma-gmwinvariant}, we can demonstrate
the following properties of conditioning-based reasoning. While a frame rule
was developed in \cite{li2023lilac}, these properties are distinct and
useful for reasoning about MPC. We omit proofs for brevity. 
\begin{lemma}
  \label{lemma-conditioning}
  Given $\pmf$, each of the following hold for $S,V_1,V_2,V_3 \in \dom(\pmf)$:
  \begin{enumerate}
    \item Given $\condp{\pmf}{S}{\detx{V_1}}$ and
      $\condp{\pmf}{V_1}{\detx{V_2}}$, then $\condp{\pmf}{S}{\detx{V_2}}$.
    \item Given $\condp{\pmf}{S}{\detx{V_1}}$ and
      $\condp{\pmf}{V_1}{\unix{V_2}}$, then $\condp{\pmf}{S}{\unix{V_2}}$.
    \item Given $\condp{\pmf}{S}{\detx{V_1}}$ and
      $\condsep{\pmf}{V_1}{V_2}{V_3}$, then $\condsep{\pmf}{S}{V_2}{V_3}$.
  \end{enumerate}
\end{lemma}

Then we can put the pieces together to prove the invariant, using automated tactics
for gate-level reasoning. We sketch some elements of this proof to focus on
novel aspects of our technique. 
\begin{proof}[Proof of Lemma \ref{lemma-gmwinvariant}]
  By induction on the length of $(\varnothing,e) \redxs (\prog,\decodegmw(E[\mesg{z}]))$.
  Encoding establishes the invariant in the basis. The most interesting inductive
  case is the $\andgmw$ gate. 
  \paragraph{Case $\andgmw$.} In this case we have:
  $$
  \begin{array}{c}
  (\varnothing,e) \redxs (\prog',\decodegmw(E[\andgmw(z,\mesg{x},\mesg{y})])) \redxs \\
    (\prog,\decodegmw(E[\mesg{z}]))
  \end{array}
  $$
  Letting $\iov(\prog) = S \cup M$, by definition and the induction hypothesis there
  exist $S^1,S^2 \subseteq S$ and $M^1,M^2 \subseteq M$
  such that:
  \begin{enumerate}[\hspace{5mm}(x.1)]
  \item $\conddetx{\progtt(\prog)}{S^1}{\{\macgv{\mesg{x}}\}}$
  \item $\condunix{\progtt(\prog)}{S^1}{M^1_C}$
  \end{enumerate}
  and also:
  \begin{enumerate}[\hspace{5mm}(y.1)]
  \item $\conddetx{\progtt(\prog)}{S^2}{\{\macgv{\mesg{y}}\}}$
  \item $\condunix{\progtt(\prog)}{S^2}{M^2_C}$
  \end{enumerate}
  and by Lemma \ref{lemma-gmwtactic} we have:
  \begin{enumerate}[\hspace{5mm}(z.1)]
  \item
    $\conddetx{\progtt(\prog)}{\{ \macgv{\mesg{x}},\macgv{\mesg{y}} \}}{\{ \macgv{\mesg{z}} \}}$
  \item $\condunix{\progtt(\prog)}{\{ \macgv{\mesg{x}},\macgv{\mesg{y}} \}}{\{ \elab{\mesg{z}}{\cid} \}}$
  \item $\condsep
    {\progtt(\prog)}
    {\{ \macgv{\mesg{x}},\macgv{\mesg{y}} \}}
    {\{ \macgv{\mesg{z}} \}}
    {\{ \elab{\mesg{z}}{1},\elab{\mesg{z}}{2} \}}$
  \end{enumerate}
  Given $x.1$ and $y.1$ and assumed uniqueness of random variables
  used in each gate we have $\conddetx{\progtt(\prog)}{S}{\{
    \macgv{\mesg{x}},\macgv{\mesg{y}}\}}$, so it follows from $z.1$
  and Lemma \ref{lemma-conditioning} (1) that
  $\conddetx{\progtt(\prog)}{S}{\{\macgv{\mesg{z}}\}}$, and also by
  $z.3$ and Lemma \ref{lemma-conditioning} (3):
  $$\condsep
  {\progtt(\prog)}
  {S}
  {\{ \macgv{\mesg{z}} \}}
  {\{ \elab{\mesg{z}}{1},\elab{\mesg{z}}{2} \}}$$
  Additionally Lemma \ref{lemma-conditioning} (2) gives $\condunix{\progtt(\prog)}{S}{\{ \elab{\mesg{z}}{\cid} \}}$,
  and $x.2$ and $y.2$ give $\condunix{\progtt(\prog)}{S}{(M^1 \cup M^2)_C}$, $z.3$ together with
  uniqueness of gate identifiers implies:
  $$
  \condunix{\progtt(\prog)}{S}{(M^1 \cup M^2)_C \cup \{  \elab{\mesg{z}}{\cid} \}}
  $$
  This implies the result.
\end{proof}
The preceding Lemma, together with some additional observations about decoding, establish correctness
of arbitrary circuits. 
\begin{theorem}
  \label{theorem-gmw}
  If $e$ is a GMW circuit protocol in $\metaprot$ with $(\varnothing,e) \redxs (\prog,())$
  then $\prog$ satisfies noninterference modulo output. 
\end{theorem}

\begin{proof}
  Given that $e$ is a GMW circuit protocol, then by definition we have:
  $$
  (\varnothing,e) \redxs (\prog',\decodegmw(\mesg{z})) \redxs (\prog,())
  $$
  where by Lemma \ref{lemma-gmwinvariant} and definition of $\decodegmw$,
  for any $H$ and $C$ letting $\iov(\prog) = S \cup M \cup P \cup O$ we
  have:
  \begin{mathpar}
    \conddetx{\progtt(\prog)}{S}{\{\macgv{\mesg{z}}\}}
    \condunix{\progtt(\prog)}{S}{M_C}
    \condsep{\progtt(\prog)}{S}{\{\macgv{\mesg{z}}\}}{\{ \elab{\mesg{z}}{1}, \elab{\mesg{z}}{2} \})}
  \end{mathpar}
  and by definition of $\decodegmw$ the value $\{\macgv{\mesg{z}}\}$ is output
  after the parties publicly reveal $\{ \elab{\mesg{z}}{1}, \elab{\mesg{z}}{2} \}$, so we have:
  \begin{mathpar} 
    \condunix{\progtt(\prog)}{S}{M_C}
    
    \condsep{\progtt(\prog)}{S}{O}{P}
  \end{mathpar}
  which implies the result.
\end{proof}

\subsection{2-Party BDOZ and Integrity Enforcement}
\label{section-example-bdoz}

\begin{fpfig}[t]{2-party BDOZ protocol library.}{fig-bdoz}
{\footnotesize{
\begin{verbatimtab}
auth(s,m,k,i) { assert(m == k + (m["delta"] * s))@i; }
  
sum_she(z,x,y,i) {
  m[z++"s"]@i := (m[x++"s"] + m[y++"s"])@i;
  m[z++"m"]@i := (m[x++"m"] + m[y++"m"])@i;
  m[z++"k"]@i := (m[x++"k"] + m[y++"k"])@i
}

open(x,i1,i2){
  m[x++"exts"]@i1 := m[x++"s"]@i2;
  m[x++"extm"]@i1 := m[x++"m"]@i2;
  auth(m[x++"exts"], m[x++"extm"], m[x++"k"], i1);
  m[x]@i1 := (m[x++"exts"] + m[x++"s"])@i1
}
\end{verbatimtab}
}}
\end{fpfig}

In a malicious setting, ``detecting cheating'' by adding
information-theoretic secure MAC codes to shares is a fundamental
technique pioneered in systems such as BDOZ and SPDZ
\cite{SPDZ1,SPDZ2,BDOZ,10.1007/978-3-030-68869-1_3}.  These protocols
assume a pre-processing phase that distributes shares with MAC codes
to clients.  This integrates well with pre-processed Beaver Triples to
implement malicious secure, and relatively efficient, multiplication
\cite{evans2018pragmatic}. Recall that Beaver triples are values $a,b,c$ with
$\fcod{a * b} = c$, unique per multiplication gate, that are secret
shared with clients during pre-processing. Here we consider the
2-party version.

A field value $v$ is secret shared among 2 clients in BDOZ in the same
manner as in GMW, but with the addition of a separate MAC values.  Each
client $\cid$ gets a pair of values $(v_\cid,m_\cid)$, where $v_\cid$
is a share of $v$ reconstructed by addition, i.e., $v = \fcod{v_1
  \fplus v_2}$, and $m_\cid$ is a MAC of $v_\cid$.  More precisely,
$m_\cid = k + (k_\Delta * v_\cid)$ where \emph{keys} $k$ and
$k_\Delta$ are known only to $\cid' \ne \cid$ and $k_\Delta$. The
\emph{local key} $k$ is unique per MAC, while the \emph{global key}
$k_\Delta$ is common to all MACs authenticated by $\cid'$. This is a
semi-homomorphic encryption scheme that supports addition of shares
and multiplication of shares by a constant-- for more details the
reader is referred to Orsini \cite{10.1007/978-3-030-68869-1_3}. We
note that while we restrict values $v$, $m$, and $k$ to the same field
in this presentation for simplicity, in general $m$ and $k$ can be in
extensions of $\mathbb{F}_p$.

We can capture both the preliminary distribution of Beaver triples and
BDOZ shares as a pre-processing predicate that establishes conditions
for initial memories (see Definition \ref{def-aprogtt}).  Here we
assume two input secrets $\elab{\secret{x}}{1}$ and
$\elab{\secret{y}}{2}$ and a single Beaver Triple to compute
$\elab{\secret{x}}{1} \ftimes \elab{\secret{y}}{2}$, but we can extend
this for additional gates.  As for GMW, we use $\macgv{\mesg{w}}$ to
refer to secret-shared values reconstructed with addition, but for
BDOZ by convention each share of $\macgv{\mesg{w}}$ is represented as
$\elab{\mesg{w\ttt{s}}}{\cid}$, the MAC of which is represented as a
$\elab{\mesg{w\ttt{m}}}{\cid}$ for all $\cid$, 
and each client holds a key $\elab{\mesg{w\ttt{k}}}{\cid}$ for
authentication of the other's share.
\begin{definition}[BDOZ preprocessing]
  Define:
  \begin{mathpar}
    \mathit{shares} \defeq
    \{ \elab{\mesg{w\ttt{s}}}{\cid}\ |\ \cid \in \{ 1, 2 \} \wedge w \in \{ a,b,c,x,y \}  \}

    \mathit{macs} \defeq  \{ \elab{\mesg{w\ttt{m}}}{\cid}\ |\ \cid \in \{ 1, 2 \} \wedge w \in \{ a,b,c,x,y \}  \}

    \mathit{keys} \defeq  \begin{array}{l}\{ \elab{\mesg{w\ttt{k}}}{\cid}\ |\ \cid \in \{ 1, 2 \} \wedge w \in
    \{ a,b,c,x,y \}  \} \cup \\ \{ \elab{\mesg{\ttt{delta}}}{\cid}\ |\ \cid \in \{ 1, 2 \} \} \end{array}
  \end{mathpar}
  Then a memory $\store$ satisfies BDOZ preprocessing iff:
  $$\dom(\store) = \{ \elab{\secret{x}}{1}, \elab{\secret{y}}{2} \} \cup \mathit{shares}
  \cup \mathit{macs} \cup \mathit{keys}$$
  and, writing $\store(\macgv{\mesg{w}})$ to denote
  $\fcod{\store(\elab{\mesg{w\ttt{s}}}{1}) + \store(\elab{\mesg{w\ttt{s}}}{2})}$,
  the following conditions hold:
  \begin{mathpar}
    \store(\macgv{\mesg{x}}) = \store(\elab{\secret{x}}{1})
    
    \store(\macgv{\mesg{y}}) = \store(\elab{\secret{y}}{2})
    
    \fcod{\store(\macgv{\mesg{a}}) * \store(\macgv{\mesg{b}})} = \store(\macgv{\mesg{c}})
  \end{mathpar}
  and for all $\cid,\cid' \in \{1,2\}$ with $\cid \ne \cid'$ and $w \in \{ a,b,c,x,y\}$:
  $$
  \store(\mx{w\ttt{m}}{\cid}) =
    \fcod{\store(\mx{w\ttt{k}}{\cid'}) + \store(\mx{\ttt{delta}}{\cid'}) * \store(\mx{w\ttt{s}}{\cid})}
  $$
\end{definition}

With these conventions, a BDOZ library is defined in Figure
\ref{fig-bdoz} and a multiplication protocol to compute $\sx{x}{1} *
\sx{y}{2}$ is defined in Figure \ref{fig-beaver}. The latter is
malicious secure assuming that client 1 is trusted. Initially, in
calls to $\ttt{sum\_she}$ each party computes the sum of their shares
of $\sx{x}{1} + \macgv{\mesg{a}}$ and $\sx{y}{2} + \macgv{\mesg{b}}$,
and non-interactively computes the associated MACs and keys in the
SHE authentication scheme. These values are then
securely opened on client 1, which authenticates client
2's share via a call to $\ttt{auth}$ that $\ttt{assert}$s
the BDOZ authentication property. Note that the
opening of $\sx{y}{2} + \macgv{\mesg{b}}$ reveals no
information to client 1 about $\sx{y}{2}$ since $\macgv{\mesg{b}}$ is
in a uniform random distribution by assumption. It is
also malicious secure since any cheating is detected
by the $\ttt{assert}$.

Following this, each party non-interactively calculates their share of
$\sx{x}{1} \ftimes \sx{y}{2}$, while client 2 calculates the MAC of
their share, and reveals both their share and MAC to client 1. These
are reveals since information about $\sx{y}{2}$ is exposed by
the multiplication.  Client 1 also calculates the appropriate key and
authenticates client 2's share following the BDOZ SHE MAC scheme
\cite{10.1007/978-3-030-68869-1_3} and the standard Beaver protocol
\cite{evans2018pragmatic}. Hence the protocol is correct and
malicious secure given trusted client 1. The authentication scheme
can also be run symmetrically on client 2.

\begin{fpfig}[t]{Authenticated 2-party multiplication with honest client 1.}{fig-beaver}
{\footnotesize
\begin{verbatimtab}
sum_she("d","a","x",1);
sum_she("d","a","x",2);
sum_she("e","b","y",1);
sum_she("e","b","y",2);

open("d",1,2);
open("e",1,2);

p["xys2"] := (m["bs"] * m["d"] + m["as"] * m["e"] + m["cs"])@2;
p["xym2"] := (m["bm"] * m["d"] + m["am"] * m["e"] + m["cm"])@2;

m["xys"]@1 := (m["bs"] * m["d"] + m["as"] * m["e"] + m["cs"] +
               m["d"] * m["e"])@1;
m["xyk"]@1 := (m["bk"] * m["d"] + m["ak"] * m["e"] + m["ck"])@1;

auth(p["xys2"], p["xym2"], m["xyk"], 1);

out@1 := m["xys"] + p["xys2"]@1
\end{verbatimtab}
}
\end{fpfig}

\subsection{Cheating Detection and Integrity}

We can carry out similar proofs of passive security for the protocol
in Figure \ref{fig-beaver} as for GMW, even using automated tactics
for the protocol in $\mathbb{F}_2$. But in the case of BDOZ we are
also concerned with malicious security. To demonstrate this, we can
show that the protocol satisfies integrity in the sense of Definition
\ref{def-integrity}. To do so, we observe that it satisfies a stronger
property, that we call cheating detection. Cheating detection says
that the adversary can only execute the protocol honestly, or or else
gets caught (and abort). Note in particular that this is \emph{not}
a form of endorsement (the dual of declassification), which allows
the adversary some leeway. Rather, it is a confirmation of high integrity.

Focusing in, we identify the \emph{adversarial inputs} as the messages
sent from the adversary to honest parties on which honest responses
to the adversary may depend. We want to say that these are the messages
that must be legitimate.
\begin{definition}
  Given $\prog$ with $\iov(\prog) = S \cup V \cup O$,
  let $X_H \subseteq \{ x \mid x \in (\houtputs \cup O_H) \wedge x \in \dom(\store) \}$.
  Then the \emph{adversarial inputs to $X_H$} is the least set
  $X_C \subseteq \cinputs$ such that $\progtt(\prog) \not\vdash X_C * X_H$.
\end{definition}
Now, we can characterize protocols with cheating detection as those where
the adversary can only behave honestly/passively, or else cause abort.
\begin{definition}[Cheating Detection]
  \emph{Cheating is detectable} in $\prog$ with $\iov(\prog) = S \cup V \cup O$ iff
  for all  $\store \in \aruns(\prog)$,
  letting $X_H = \{ x \mid x \in (\houtputs \cup O_H) \wedge x \in \dom(\store) \}$,
  and letting $X_C$ be the adversarial inputs to $X_H$,
  there exists $\sigma'\in \runs(\prog)$
  with $\store_{X_C} = \store'_{X_C}$.  
\end{definition}

It is straightforward to demonstrate that cheating detection has integrity,
since only the ``passive'' adversary can elicit a response from honest parties. 
\begin{lemma}
  \label{lemma-cheating}
  If cheating is detectable in $\prog$, then $\prog$ has integrity.
\end{lemma}

In the case of BDOZ, cheating detection is accomplished by the information-theoretic
security of the encryption scheme \cite{evans2018pragmatic}. Furthermore, the symmetry of
the protocol in Figure \ref{fig-beaver} ensures that both parties will authenticate
shares, so it is robust to corruption of either party.

\section{Conclusion and Future Work}

The $\metaprot/\minifed$ language model developed in this paper, along
with the suite of hyperproperties explored, are a foundation
for defining low-level MPC protocols and exploring methods for security
verification.  There are several promising avenues for future
work. Concurrency will be important to capture common MPC idioms such
as commitment and circuit optimizations, and to consider the UC
security model \cite{evans2018pragmatic,viaduct-UC}. To support better
verification, integration of our tactics with a proof assistant is a
compelling direction that resonate with previous work \cite{8429300}.

Better automated verification is a main goal for future work.
Security types have been effectively leveraged for high-level
programs, for both passive and malicious security. This suggests that
they can be similarly applied at low levels in a way that is
complementary to high-level systems. Generalizing our automated
analysis in $\mathbb{F}_2$ to larger arithmetic fields is also a
compelling technical challenge with potential practical benefits,
e.g., application to arithmetic circuits.

\bibliographystyle{ACM-Reference-Format}
\bibliography{logic-bibliography,secure-computation-bibliography}


\begin{thebibliography}{46}


\ifx \showCODEN    \undefined \def \showCODEN     #1{\unskip}     \fi
\ifx \showDOI      \undefined \def \showDOI       #1{#1}\fi
\ifx \showISBNx    \undefined \def \showISBNx     #1{\unskip}     \fi
\ifx \showISBNxiii \undefined \def \showISBNxiii  #1{\unskip}     \fi
\ifx \showISSN     \undefined \def \showISSN      #1{\unskip}     \fi
\ifx \showLCCN     \undefined \def \showLCCN      #1{\unskip}     \fi
\ifx \shownote     \undefined \def \shownote      #1{#1}          \fi
\ifx \showarticletitle \undefined \def \showarticletitle #1{#1}   \fi
\ifx \showURL      \undefined \def \showURL       {\relax}        \fi
\providecommand\bibfield[2]{#2}
\providecommand\bibinfo[2]{#2}
\providecommand\natexlab[1]{#1}
\providecommand\showeprint[2][]{arXiv:#2}

\bibitem[\protect\citeauthoryear{Abadi and Rogaway}{Abadi and Rogaway}{2000}]%
        {10.1007/3-540-44929-9_1}
\bibfield{author}{\bibinfo{person}{Martin Abadi} {and} \bibinfo{person}{Phillip
  Rogaway}.} \bibinfo{year}{2000}\natexlab{}.
\newblock \showarticletitle{Reconciling Two Views of Cryptography}. In
  \bibinfo{booktitle}{\emph{Theoretical Computer Science: Exploring New
  Frontiers of Theoretical Informatics}},
  \bibfield{editor}{\bibinfo{person}{Jan van Leeuwen}, \bibinfo{person}{Osamu
  Watanabe}, \bibinfo{person}{Masami Hagiya}, \bibinfo{person}{Peter~D.
  Mosses}, {and} \bibinfo{person}{Takayasu Ito}} (Eds.).
  \bibinfo{publisher}{Springer Berlin Heidelberg}, \bibinfo{address}{Berlin,
  Heidelberg}, \bibinfo{pages}{3--22}.
\newblock
\showISBNx{978-3-540-44929-4}


\bibitem[\protect\citeauthoryear{Acay, Gancher, Recto, and Myers}{Acay
  et~al\mbox{.}}{2024}]%
        {viaduct-UC}
\bibfield{author}{\bibinfo{person}{C. Acay}, \bibinfo{person}{J. Gancher},
  \bibinfo{person}{R. Recto}, {and} \bibinfo{person}{A. Myers}.}
  \bibinfo{year}{2024}\natexlab{}.
\newblock \showarticletitle{Secure Synthesis of Distributed Cryptographic
  Applications}. In \bibinfo{booktitle}{\emph{2024 IEEE 37th Computer Security
  Foundations Symposium (CSF)}}. \bibinfo{publisher}{IEEE Computer Society},
  \bibinfo{address}{Los Alamitos, CA, USA}, \bibinfo{pages}{315--330}.
\newblock
\showISSN{2374-8303}
\urldef\tempurl%
\url{https://doi.org/10.1109/CSF61375.2024.00021}
\showDOI{\tempurl}


\bibitem[\protect\citeauthoryear{Acay, Recto, Gancher, Myers, and Shi}{Acay
  et~al\mbox{.}}{2021}]%
        {10.1145/3453483.3454074}
\bibfield{author}{\bibinfo{person}{Co\c{s}ku Acay}, \bibinfo{person}{Rolph
  Recto}, \bibinfo{person}{Joshua Gancher}, \bibinfo{person}{Andrew~C. Myers},
  {and} \bibinfo{person}{Elaine Shi}.} \bibinfo{year}{2021}\natexlab{}.
\newblock \showarticletitle{Viaduct: An Extensible, Optimizing Compiler for
  Secure Distributed Programs}. In \bibinfo{booktitle}{\emph{Proceedings of the
  42nd ACM SIGPLAN International Conference on Programming Language Design and
  Implementation}} (Virtual, Canada) \emph{(\bibinfo{series}{PLDI 2021})}.
  \bibinfo{publisher}{Association for Computing Machinery},
  \bibinfo{address}{New York, NY, USA}, \bibinfo{pages}{740--755}.
\newblock
\showISBNx{9781450383912}
\urldef\tempurl%
\url{https://doi.org/10.1145/3453483.3454074}
\showDOI{\tempurl}


\bibitem[\protect\citeauthoryear{Albarghouthi, D'Antoni, Drews, and
  Nori}{Albarghouthi et~al\mbox{.}}{2017}]%
        {albarghouthi2017fairsquare}
\bibfield{author}{\bibinfo{person}{Aws Albarghouthi}, \bibinfo{person}{Loris
  D'Antoni}, \bibinfo{person}{Samuel Drews}, {and} \bibinfo{person}{Aditya~V
  Nori}.} \bibinfo{year}{2017}\natexlab{}.
\newblock \showarticletitle{Fairsquare: probabilistic verification of program
  fairness}.
\newblock \bibinfo{journal}{\emph{Proceedings of the ACM on Programming
  Languages}} \bibinfo{volume}{1}, \bibinfo{number}{OOPSLA}
  (\bibinfo{year}{2017}), \bibinfo{pages}{1--30}.
\newblock


\bibitem[\protect\citeauthoryear{Almeida, Barbosa, Barthe, Pacheco, Pereira,
  and Portela}{Almeida et~al\mbox{.}}{2018}]%
        {almeida2018enforcing}
\bibfield{author}{\bibinfo{person}{Jos{\'e}~Bacelar Almeida},
  \bibinfo{person}{Manuel Barbosa}, \bibinfo{person}{Gilles Barthe},
  \bibinfo{person}{Hugo Pacheco}, \bibinfo{person}{Vitor Pereira}, {and}
  \bibinfo{person}{Bernardo Portela}.} \bibinfo{year}{2018}\natexlab{}.
\newblock \showarticletitle{Enforcing ideal-world leakage bounds in real-world
  secret sharing MPC frameworks}. In \bibinfo{booktitle}{\emph{2018 IEEE 31st
  Computer Security Foundations Symposium (CSF)}}. IEEE,
  \bibinfo{pages}{132--146}.
\newblock


\bibitem[\protect\citeauthoryear{Askarov and Sabelfeld}{Askarov and
  Sabelfeld}{2007}]%
        {4223226}
\bibfield{author}{\bibinfo{person}{Aslan Askarov} {and} \bibinfo{person}{Andrei
  Sabelfeld}.} \bibinfo{year}{2007}\natexlab{}.
\newblock \showarticletitle{Gradual Release: Unifying Declassification,
  Encryption and Key Release Policies}. In \bibinfo{booktitle}{\emph{2007 IEEE
  Symposium on Security and Privacy (SP '07)}}. \bibinfo{pages}{207--221}.
\newblock
\urldef\tempurl%
\url{https://doi.org/10.1109/SP.2007.22}
\showDOI{\tempurl}


\bibitem[\protect\citeauthoryear{Aspis}{Aspis}{2018}]%
        {aspis2018linear}
\bibfield{author}{\bibinfo{person}{Yaniv Aspis}.}
  \bibinfo{year}{2018}\natexlab{}.
\newblock \emph{\bibinfo{title}{A Linear Algebraic Approach to Logic
  Programming}}.
\newblock \bibinfo{thesistype}{Ph.D. Dissertation}. \bibinfo{school}{Master
  thesis at Imperial College London}.
\newblock


\bibitem[\protect\citeauthoryear{Barthe, Hsu, and Liao}{Barthe
  et~al\mbox{.}}{2019}]%
        {barthe2019probabilistic}
\bibfield{author}{\bibinfo{person}{Gilles Barthe}, \bibinfo{person}{Justin
  Hsu}, {and} \bibinfo{person}{Kevin Liao}.} \bibinfo{year}{2019}\natexlab{}.
\newblock \showarticletitle{A probabilistic separation logic}.
\newblock \bibinfo{journal}{\emph{Proc. ACM Program. Lang.}}
  \bibinfo{volume}{4}, \bibinfo{number}{POPL}, Article \bibinfo{articleno}{55}
  (\bibinfo{date}{dec} \bibinfo{year}{2019}), \bibinfo{numpages}{30}~pages.
\newblock
\urldef\tempurl%
\url{https://doi.org/10.1145/3371123}
\showDOI{\tempurl}


\bibitem[\protect\citeauthoryear{Bendlin, Damg{\aa}rd, Orlandi, and
  Zakarias}{Bendlin et~al\mbox{.}}{2011}]%
        {BDOZ}
\bibfield{author}{\bibinfo{person}{Rikke Bendlin}, \bibinfo{person}{Ivan
  Damg{\aa}rd}, \bibinfo{person}{Claudio Orlandi}, {and} \bibinfo{person}{Sarah
  Zakarias}.} \bibinfo{year}{2011}\natexlab{}.
\newblock \showarticletitle{Semi-homomorphic Encryption and Multiparty
  Computation}. In \bibinfo{booktitle}{\emph{Advances in Cryptology --
  EUROCRYPT 2011}}, \bibfield{editor}{\bibinfo{person}{Kenneth~G. Paterson}}
  (Ed.). \bibinfo{publisher}{Springer Berlin Heidelberg},
  \bibinfo{address}{Berlin, Heidelberg}, \bibinfo{pages}{169--188}.
\newblock
\showISBNx{978-3-642-20465-4}


\bibitem[\protect\citeauthoryear{Bingham, Chen, Jankowiak, Obermeyer, Pradhan,
  Karaletsos, Singh, Szerlip, Horsfall, and Goodman}{Bingham
  et~al\mbox{.}}{2019}]%
        {bingham2019pyro}
\bibfield{author}{\bibinfo{person}{Eli Bingham}, \bibinfo{person}{Jonathan~P
  Chen}, \bibinfo{person}{Martin Jankowiak}, \bibinfo{person}{Fritz Obermeyer},
  \bibinfo{person}{Neeraj Pradhan}, \bibinfo{person}{Theofanis Karaletsos},
  \bibinfo{person}{Rohit Singh}, \bibinfo{person}{Paul Szerlip},
  \bibinfo{person}{Paul Horsfall}, {and} \bibinfo{person}{Noah~D Goodman}.}
  \bibinfo{year}{2019}\natexlab{}.
\newblock \showarticletitle{Pyro: Deep universal probabilistic programming}.
\newblock \bibinfo{journal}{\emph{The Journal of Machine Learning Research}}
  \bibinfo{volume}{20}, \bibinfo{number}{1} (\bibinfo{year}{2019}),
  \bibinfo{pages}{973--978}.
\newblock


\bibitem[\protect\citeauthoryear{Bogdanov, Laud, and Randmets}{Bogdanov
  et~al\mbox{.}}{2014}]%
        {10.1145/2637113.2637119}
\bibfield{author}{\bibinfo{person}{Dan Bogdanov}, \bibinfo{person}{Peeter
  Laud}, {and} \bibinfo{person}{Jaak Randmets}.}
  \bibinfo{year}{2014}\natexlab{}.
\newblock \showarticletitle{Domain-Polymorphic Programming of
  Privacy-Preserving Applications}. In \bibinfo{booktitle}{\emph{Proceedings of
  the Ninth Workshop on Programming Languages and Analysis for Security}}
  (Uppsala, Sweden) \emph{(\bibinfo{series}{PLAS'14})}.
  \bibinfo{publisher}{Association for Computing Machinery},
  \bibinfo{address}{New York, NY, USA}, \bibinfo{pages}{53--65}.
\newblock
\showISBNx{9781450328623}
\urldef\tempurl%
\url{https://doi.org/10.1145/2637113.2637119}
\showDOI{\tempurl}


\bibitem[\protect\citeauthoryear{Brzuska and Oechsner}{Brzuska and
  Oechsner}{2023}]%
        {5a51987acaa84c43bb4bf5bcc7d01683}
\bibfield{author}{\bibinfo{person}{Chris Brzuska} {and} \bibinfo{person}{Sabine
  Oechsner}.} \bibinfo{year}{2023}\natexlab{}.
\newblock \showarticletitle{A State-Separating Proof for Yao's Garbling
  Scheme}. In \bibinfo{booktitle}{\emph{Proceedings - 2023 IEEE 36th Computer
  Security Foundations Symposium, CSF 2023}}
  \emph{(\bibinfo{series}{Proceedings - IEEE Computer Security Foundations
  Symposium})}. \bibinfo{publisher}{IEEE Computer Society},
  \bibinfo{pages}{137--152}.
\newblock
\urldef\tempurl%
\url{https://doi.org/10.1109/CSF57540.2023.00009}
\showDOI{\tempurl}
\newblock
\shownote{36th IEEE Computer Security Foundations Symposium, CSF 2023 ;
  Conference date: 09-07-2023 Through 13-07-2023.}


\bibitem[\protect\citeauthoryear{Carpenter, Gelman, Hoffman, Lee, Goodrich,
  Betancourt, Brubaker, Guo, Li, and Riddell}{Carpenter et~al\mbox{.}}{2017}]%
        {carpenter2017stan}
\bibfield{author}{\bibinfo{person}{Bob Carpenter}, \bibinfo{person}{Andrew
  Gelman}, \bibinfo{person}{Matthew~D Hoffman}, \bibinfo{person}{Daniel Lee},
  \bibinfo{person}{Ben Goodrich}, \bibinfo{person}{Michael Betancourt},
  \bibinfo{person}{Marcus~A Brubaker}, \bibinfo{person}{Jiqiang Guo},
  \bibinfo{person}{Peter Li}, {and} \bibinfo{person}{Allen Riddell}.}
  \bibinfo{year}{2017}\natexlab{}.
\newblock \showarticletitle{Stan: A probabilistic programming language}.
\newblock \bibinfo{journal}{\emph{Journal of statistical software}}
  \bibinfo{volume}{76} (\bibinfo{year}{2017}).
\newblock


\bibitem[\protect\citeauthoryear{Clarkson and Schneider}{Clarkson and
  Schneider}{2010}]%
        {10.5555/1891823.1891830}
\bibfield{author}{\bibinfo{person}{Michael~R. Clarkson} {and}
  \bibinfo{person}{Fred~B. Schneider}.} \bibinfo{year}{2010}\natexlab{}.
\newblock \showarticletitle{Hyperproperties}.
\newblock \bibinfo{journal}{\emph{J. Comput. Secur.}} \bibinfo{volume}{18},
  \bibinfo{number}{6} (\bibinfo{date}{sep} \bibinfo{year}{2010}),
  \bibinfo{pages}{1157--1210}.
\newblock
\showISSN{0926-227X}


\bibitem[\protect\citeauthoryear{Damg{\aa}rd and Orlandi}{Damg{\aa}rd and
  Orlandi}{2010}]%
        {SPDZ1}
\bibfield{author}{\bibinfo{person}{Ivan Damg{\aa}rd} {and}
  \bibinfo{person}{Claudio Orlandi}.} \bibinfo{year}{2010}\natexlab{}.
\newblock \showarticletitle{Multiparty Computation for Dishonest Majority: From
  Passive to Active Security at Low Cost}. In
  \bibinfo{booktitle}{\emph{Advances in Cryptology -- CRYPTO 2010}},
  \bibfield{editor}{\bibinfo{person}{Tal Rabin}} (Ed.).
  \bibinfo{publisher}{Springer Berlin Heidelberg}, \bibinfo{address}{Berlin,
  Heidelberg}, \bibinfo{pages}{558--576}.
\newblock
\showISBNx{978-3-642-14623-7}


\bibitem[\protect\citeauthoryear{Damg{\aa}rd, Pastro, Smart, and
  Zakarias}{Damg{\aa}rd et~al\mbox{.}}{2012}]%
        {SPDZ2}
\bibfield{author}{\bibinfo{person}{Ivan Damg{\aa}rd}, \bibinfo{person}{Valerio
  Pastro}, \bibinfo{person}{Nigel Smart}, {and} \bibinfo{person}{Sarah
  Zakarias}.} \bibinfo{year}{2012}\natexlab{}.
\newblock \showarticletitle{Multiparty Computation from Somewhat Homomorphic
  Encryption}. In \bibinfo{booktitle}{\emph{Advances in Cryptology -- CRYPTO
  2012}}, \bibfield{editor}{\bibinfo{person}{Reihaneh Safavi-Naini} {and}
  \bibinfo{person}{Ran Canetti}} (Eds.). \bibinfo{publisher}{Springer Berlin
  Heidelberg}, \bibinfo{address}{Berlin, Heidelberg},
  \bibinfo{pages}{643--662}.
\newblock
\showISBNx{978-3-642-32009-5}


\bibitem[\protect\citeauthoryear{Darais, Sweet, Liu, and Hicks}{Darais
  et~al\mbox{.}}{2019}]%
        {darais2019language}
\bibfield{author}{\bibinfo{person}{David Darais}, \bibinfo{person}{Ian Sweet},
  \bibinfo{person}{Chang Liu}, {and} \bibinfo{person}{Michael Hicks}.}
  \bibinfo{year}{2019}\natexlab{}.
\newblock \showarticletitle{A language for probabilistically oblivious
  computation}.
\newblock \bibinfo{journal}{\emph{Proceedings of the ACM on Programming
  Languages}} \bibinfo{volume}{4}, \bibinfo{number}{POPL}
  (\bibinfo{year}{2019}), \bibinfo{pages}{1--31}.
\newblock


\bibitem[\protect\citeauthoryear{De~Raedt, Kimmig, and Toivonen}{De~Raedt
  et~al\mbox{.}}{2007}]%
        {de2007problog}
\bibfield{author}{\bibinfo{person}{Luc De~Raedt}, \bibinfo{person}{Angelika
  Kimmig}, {and} \bibinfo{person}{Hannu Toivonen}.}
  \bibinfo{year}{2007}\natexlab{}.
\newblock \showarticletitle{ProbLog: A probabilistic Prolog and its application
  in link discovery}. In \bibinfo{booktitle}{\emph{IJCAI 2007, Proceedings of
  the 20th international joint conference on artificial intelligence}}.
  IJCAI-INT JOINT CONF ARTIF INTELL, \bibinfo{pages}{2462--2467}.
\newblock


\bibitem[\protect\citeauthoryear{Evans, Kolesnikov, Rosulek,
  et~al\mbox{.}}{Evans et~al\mbox{.}}{2018}]%
        {evans2018pragmatic}
\bibfield{author}{\bibinfo{person}{David Evans}, \bibinfo{person}{Vladimir
  Kolesnikov}, \bibinfo{person}{Mike Rosulek}, {et~al\mbox{.}}}
  \bibinfo{year}{2018}\natexlab{}.
\newblock \showarticletitle{A pragmatic introduction to secure multi-party
  computation}.
\newblock \bibinfo{journal}{\emph{Foundations and Trends{\textregistered} in
  Privacy and Security}} \bibinfo{volume}{2}, \bibinfo{number}{2-3}
  (\bibinfo{year}{2018}), \bibinfo{pages}{70--246}.
\newblock


\bibitem[\protect\citeauthoryear{Gao, Peng, Tan, Zheng, and Xiao}{Gao
  et~al\mbox{.}}{2022}]%
        {gao2022symmeproof}
\bibfield{author}{\bibinfo{person}{Shang Gao}, \bibinfo{person}{Zhe Peng},
  \bibinfo{person}{Feng Tan}, \bibinfo{person}{Yuanqing Zheng}, {and}
  \bibinfo{person}{Bin Xiao}.} \bibinfo{year}{2022}\natexlab{}.
\newblock \showarticletitle{SymmeProof: Compact zero-knowledge argument for
  blockchain confidential transactions}.
\newblock \bibinfo{journal}{\emph{IEEE Transactions on Dependable and Secure
  Computing}} (\bibinfo{year}{2022}).
\newblock


\bibitem[\protect\citeauthoryear{Goldreich, Micali, and Wigderson}{Goldreich
  et~al\mbox{.}}{2019}]%
        {goldreich2019play}
\bibfield{author}{\bibinfo{person}{Oded Goldreich}, \bibinfo{person}{Silvio
  Micali}, {and} \bibinfo{person}{Avi Wigderson}.}
  \bibinfo{year}{2019}\natexlab{}.
\newblock \showarticletitle{How to play any mental game, or a completeness
  theorem for protocols with honest majority}.
\newblock In \bibinfo{booktitle}{\emph{Providing Sound Foundations for
  Cryptography: On the Work of Shafi Goldwasser and Silvio Micali}}.
  \bibinfo{pages}{307--328}.
\newblock


\bibitem[\protect\citeauthoryear{Haagh, Karbyshev, Oechsner, Spitters, and
  Strub}{Haagh et~al\mbox{.}}{2018}]%
        {8429300}
\bibfield{author}{\bibinfo{person}{H. Haagh}, \bibinfo{person}{A. Karbyshev},
  \bibinfo{person}{S. Oechsner}, \bibinfo{person}{B. Spitters}, {and}
  \bibinfo{person}{P. Strub}.} \bibinfo{year}{2018}\natexlab{}.
\newblock \showarticletitle{Computer-Aided Proofs for Multiparty Computation
  with Active Security}. In \bibinfo{booktitle}{\emph{2018 IEEE 31st Computer
  Security Foundations Symposium (CSF)}}. \bibinfo{publisher}{IEEE Computer
  Society}, \bibinfo{address}{Los Alamitos, CA, USA},
  \bibinfo{pages}{119--131}.
\newblock
\showISSN{2374-8303}
\urldef\tempurl%
\url{https://doi.org/10.1109/CSF.2018.00016}
\showDOI{\tempurl}


\bibitem[\protect\citeauthoryear{Holtzen, Van~den Broeck, and
  Millstein}{Holtzen et~al\mbox{.}}{2020}]%
        {holtzen2020scaling}
\bibfield{author}{\bibinfo{person}{Steven Holtzen}, \bibinfo{person}{Guy
  Van~den Broeck}, {and} \bibinfo{person}{Todd Millstein}.}
  \bibinfo{year}{2020}\natexlab{}.
\newblock \showarticletitle{Scaling exact inference for discrete probabilistic
  programs}.
\newblock \bibinfo{journal}{\emph{Proceedings of the ACM on Programming
  Languages}} \bibinfo{volume}{4}, \bibinfo{number}{OOPSLA}
  (\bibinfo{year}{2020}), \bibinfo{pages}{1--31}.
\newblock


\bibitem[\protect\citeauthoryear{Hunt, Sands, and Stucki}{Hunt
  et~al\mbox{.}}{2023}]%
        {10.1145/3571740}
\bibfield{author}{\bibinfo{person}{Sebastian Hunt}, \bibinfo{person}{David
  Sands}, {and} \bibinfo{person}{Sandro Stucki}.}
  \bibinfo{year}{2023}\natexlab{}.
\newblock \showarticletitle{Reconciling Shannon and Scott with a Lattice of
  Computable Information}.
\newblock \bibinfo{journal}{\emph{Proc. ACM Program. Lang.}}
  \bibinfo{volume}{7}, \bibinfo{number}{POPL}, Article \bibinfo{articleno}{68}
  (\bibinfo{date}{jan} \bibinfo{year}{2023}), \bibinfo{numpages}{30}~pages.
\newblock
\urldef\tempurl%
\url{https://doi.org/10.1145/3571740}
\showDOI{\tempurl}


\bibitem[\protect\citeauthoryear{Ishai, Kushilevitz, Ostrovsky, and
  Sahai}{Ishai et~al\mbox{.}}{2009}]%
        {ishai2009zero}
\bibfield{author}{\bibinfo{person}{Yuval Ishai}, \bibinfo{person}{Eyal
  Kushilevitz}, \bibinfo{person}{Rafail Ostrovsky}, {and} \bibinfo{person}{Amit
  Sahai}.} \bibinfo{year}{2009}\natexlab{}.
\newblock \showarticletitle{Zero-knowledge proofs from secure multiparty
  computation}.
\newblock \bibinfo{journal}{\emph{SIAM J. Comput.}} \bibinfo{volume}{39},
  \bibinfo{number}{3} (\bibinfo{year}{2009}), \bibinfo{pages}{1121--1152}.
\newblock


\bibitem[\protect\citeauthoryear{Knott, Venkataraman, Hannun, Sengupta,
  Ibrahim, and van~der Maaten}{Knott et~al\mbox{.}}{2021}]%
        {knott2021crypten}
\bibfield{author}{\bibinfo{person}{Brian Knott}, \bibinfo{person}{Shobha
  Venkataraman}, \bibinfo{person}{Awni Hannun}, \bibinfo{person}{Shubho
  Sengupta}, \bibinfo{person}{Mark Ibrahim}, {and} \bibinfo{person}{Laurens
  van~der Maaten}.} \bibinfo{year}{2021}\natexlab{}.
\newblock \showarticletitle{Crypten: Secure multi-party computation meets
  machine learning}.
\newblock \bibinfo{journal}{\emph{Advances in Neural Information Processing
  Systems}}  \bibinfo{volume}{34} (\bibinfo{year}{2021}),
  \bibinfo{pages}{4961--4973}.
\newblock


\bibitem[\protect\citeauthoryear{Koch, Krenn, Pellegrino, and Ramacher}{Koch
  et~al\mbox{.}}{2020}]%
        {koch2020privacy}
\bibfield{author}{\bibinfo{person}{Karl Koch}, \bibinfo{person}{Stephan Krenn},
  \bibinfo{person}{Donato Pellegrino}, {and} \bibinfo{person}{Sebastian
  Ramacher}.} \bibinfo{year}{2020}\natexlab{}.
\newblock \showarticletitle{Privacy-preserving analytics for data markets using
  MPC}.
\newblock In \bibinfo{booktitle}{\emph{IFIP International Summer School on
  Privacy and Identity Management}}. \bibinfo{publisher}{Springer},
  \bibinfo{pages}{226--246}.
\newblock


\bibitem[\protect\citeauthoryear{Li, Ahmed, and Holtzen}{Li
  et~al\mbox{.}}{2023}]%
        {li2023lilac}
\bibfield{author}{\bibinfo{person}{John~M Li}, \bibinfo{person}{Amal Ahmed},
  {and} \bibinfo{person}{Steven Holtzen}.} \bibinfo{year}{2023}\natexlab{}.
\newblock \showarticletitle{Lilac: a Modal Separation Logic for Conditional
  Probability}.
\newblock \bibinfo{journal}{\emph{Proceedings of the ACM on Programming
  Languages}} \bibinfo{volume}{7}, \bibinfo{number}{PLDI}
  (\bibinfo{year}{2023}), \bibinfo{pages}{148--171}.
\newblock


\bibitem[\protect\citeauthoryear{Li, Dowsley, and De~Cock}{Li
  et~al\mbox{.}}{2021}]%
        {li2021privacy}
\bibfield{author}{\bibinfo{person}{Xiling Li}, \bibinfo{person}{Rafael
  Dowsley}, {and} \bibinfo{person}{Martine De~Cock}.}
  \bibinfo{year}{2021}\natexlab{}.
\newblock \showarticletitle{Privacy-preserving feature selection with secure
  multiparty computation}. In \bibinfo{booktitle}{\emph{International
  Conference on Machine Learning}}. PMLR, \bibinfo{pages}{6326--6336}.
\newblock


\bibitem[\protect\citeauthoryear{Lindell}{Lindell}{2017}]%
        {Lindell2017}
\bibfield{author}{\bibinfo{person}{Yehuda Lindell}.}
  \bibinfo{year}{2017}\natexlab{}.
\newblock \bibinfo{booktitle}{\emph{How to Simulate It -- A Tutorial on the
  Simulation Proof Technique}}.
\newblock \bibinfo{publisher}{Springer International Publishing},
  \bibinfo{address}{Cham}, \bibinfo{pages}{277--346}.
\newblock
\showISBNx{978-3-319-57048-8}
\urldef\tempurl%
\url{https://doi.org/10.1007/978-3-319-57048-8_6}
\showDOI{\tempurl}


\bibitem[\protect\citeauthoryear{Liu, Tian, Zhou, Xiao, and Ansari}{Liu
  et~al\mbox{.}}{2020}]%
        {liu2020privacy}
\bibfield{author}{\bibinfo{person}{Jun Liu}, \bibinfo{person}{Yuan Tian},
  \bibinfo{person}{Yu Zhou}, \bibinfo{person}{Yang Xiao}, {and}
  \bibinfo{person}{Nirwan Ansari}.} \bibinfo{year}{2020}\natexlab{}.
\newblock \showarticletitle{Privacy preserving distributed data mining based on
  secure multi-party computation}.
\newblock \bibinfo{journal}{\emph{Computer Communications}}
  \bibinfo{volume}{153} (\bibinfo{year}{2020}), \bibinfo{pages}{208--216}.
\newblock


\bibitem[\protect\citeauthoryear{Lu, Yurek, Kulshreshtha, Govind, Kate, and
  Miller}{Lu et~al\mbox{.}}{2019}]%
        {lu2019honeybadgermpc}
\bibfield{author}{\bibinfo{person}{Donghang Lu}, \bibinfo{person}{Thomas
  Yurek}, \bibinfo{person}{Samarth Kulshreshtha}, \bibinfo{person}{Rahul
  Govind}, \bibinfo{person}{Aniket Kate}, {and} \bibinfo{person}{Andrew
  Miller}.} \bibinfo{year}{2019}\natexlab{}.
\newblock \showarticletitle{Honeybadgermpc and asynchromix: Practical
  asynchronous mpc and its application to anonymous communication}. In
  \bibinfo{booktitle}{\emph{Proceedings of the 2019 ACM SIGSAC Conference on
  Computer and Communications Security}}. \bibinfo{pages}{887--903}.
\newblock


\bibitem[\protect\citeauthoryear{Malkhi, Nisan, Pinkas, and Sella}{Malkhi
  et~al\mbox{.}}{2004}]%
        {269581}
\bibfield{author}{\bibinfo{person}{Dahlia Malkhi}, \bibinfo{person}{Noam
  Nisan}, \bibinfo{person}{Benny Pinkas}, {and} \bibinfo{person}{Yaron Sella}.}
  \bibinfo{year}{2004}\natexlab{}.
\newblock \showarticletitle{{Fairplay{\textemdash}A} Secure {Two-Party}
  Computation System}. In \bibinfo{booktitle}{\emph{13th USENIX Security
  Symposium (USENIX Security 04)}}. \bibinfo{publisher}{USENIX Association},
  \bibinfo{address}{San Diego, CA}.


\bibitem[\protect\citeauthoryear{Mitchell, Sharma, Stefan, and
  Zimmerman}{Mitchell et~al\mbox{.}}{2012}]%
        {6266151}
\bibfield{author}{\bibinfo{person}{John~C. Mitchell}, \bibinfo{person}{Rahul
  Sharma}, \bibinfo{person}{Deian Stefan}, {and} \bibinfo{person}{Joe
  Zimmerman}.} \bibinfo{year}{2012}\natexlab{}.
\newblock \showarticletitle{Information-Flow Control for Programming on
  Encrypted Data}. In \bibinfo{booktitle}{\emph{2012 IEEE 25th Computer
  Security Foundations Symposium}}. \bibinfo{pages}{45--60}.
\newblock
\urldef\tempurl%
\url{https://doi.org/10.1109/CSF.2012.30}
\showDOI{\tempurl}


\bibitem[\protect\citeauthoryear{Orsini}{Orsini}{2021}]%
        {10.1007/978-3-030-68869-1_3}
\bibfield{author}{\bibinfo{person}{Emmanuela Orsini}.}
  \bibinfo{year}{2021}\natexlab{}.
\newblock \showarticletitle{Efficient, Actively Secure MPC with a Dishonest
  Majority: A Survey}. In \bibinfo{booktitle}{\emph{Arithmetic of Finite
  Fields}}, \bibfield{editor}{\bibinfo{person}{Jean~Claude Bajard} {and}
  \bibinfo{person}{Alev Topuzo{\u{g}}lu}} (Eds.). \bibinfo{publisher}{Springer
  International Publishing}, \bibinfo{address}{Cham}, \bibinfo{pages}{42--71}.
\newblock
\showISBNx{978-3-030-68869-1}


\bibitem[\protect\citeauthoryear{Pfeffer}{Pfeffer}{2009}]%
        {pfeffer2009figaro}
\bibfield{author}{\bibinfo{person}{Avi Pfeffer}.}
  \bibinfo{year}{2009}\natexlab{}.
\newblock \showarticletitle{Figaro: An object-oriented probabilistic
  programming language}.
\newblock \bibinfo{journal}{\emph{Charles River Analytics Technical Report}}
  \bibinfo{volume}{137}, \bibinfo{number}{96} (\bibinfo{year}{2009}),
  \bibinfo{pages}{4}.
\newblock


\bibitem[\protect\citeauthoryear{Rastogi, Hammer, and Hicks}{Rastogi
  et~al\mbox{.}}{2014}]%
        {rastogi2014wysteria}
\bibfield{author}{\bibinfo{person}{Aseem Rastogi}, \bibinfo{person}{Matthew~A
  Hammer}, {and} \bibinfo{person}{Michael Hicks}.}
  \bibinfo{year}{2014}\natexlab{}.
\newblock \showarticletitle{Wysteria: A programming language for generic,
  mixed-mode multiparty computations}. In \bibinfo{booktitle}{\emph{2014 IEEE
  Symposium on Security and Privacy}}. IEEE, \bibinfo{pages}{655--670}.
\newblock


\bibitem[\protect\citeauthoryear{Rastogi, Swamy, and Hicks}{Rastogi
  et~al\mbox{.}}{2019}]%
        {wysstar}
\bibfield{author}{\bibinfo{person}{Aseem Rastogi}, \bibinfo{person}{Nikhil
  Swamy}, {and} \bibinfo{person}{Michael Hicks}.}
  \bibinfo{year}{2019}\natexlab{}.
\newblock \showarticletitle{Wys*: {A} {DSL} for Verified Secure Multi-party
  Computations}. In \bibinfo{booktitle}{\emph{8th International Conference on
  Principles of Security and Trust (POST)}} \emph{(\bibinfo{series}{Lecture
  Notes in Computer Science})}, \bibfield{editor}{\bibinfo{person}{Flemming
  Nielson} {and} \bibinfo{person}{David Sands}} (Eds.),
  Vol.~\bibinfo{volume}{11426}. \bibinfo{publisher}{Springer},
  \bibinfo{pages}{99--122}.
\newblock
\urldef\tempurl%
\url{https://doi.org/10.1007/978-3-030-17138-4\_5}
\showDOI{\tempurl}


\bibitem[\protect\citeauthoryear{Saad, Rinard, and Mansinghka}{Saad
  et~al\mbox{.}}{2021}]%
        {saad2021sppl}
\bibfield{author}{\bibinfo{person}{Feras~A Saad}, \bibinfo{person}{Martin~C
  Rinard}, {and} \bibinfo{person}{Vikash~K Mansinghka}.}
  \bibinfo{year}{2021}\natexlab{}.
\newblock \showarticletitle{SPPL: probabilistic programming with fast exact
  symbolic inference}. In \bibinfo{booktitle}{\emph{Proceedings of the 42nd ACM
  SIGPLAN International Conference on Programming Language Design and
  Implementation}}. \bibinfo{pages}{804--819}.
\newblock


\bibitem[\protect\citeauthoryear{Sabelfeld and Sands}{Sabelfeld and
  Sands}{2009}]%
        {sabelfeld2009declassification}
\bibfield{author}{\bibinfo{person}{Andrei Sabelfeld} {and}
  \bibinfo{person}{David Sands}.} \bibinfo{year}{2009}\natexlab{}.
\newblock \showarticletitle{Declassification: Dimensions and Principles}.
\newblock \bibinfo{journal}{\emph{J. Comput. Secur.}} \bibinfo{volume}{17},
  \bibinfo{number}{5} (\bibinfo{date}{oct} \bibinfo{year}{2009}),
  \bibinfo{pages}{517--548}.
\newblock
\showISSN{0926-227X}


\bibitem[\protect\citeauthoryear{Sakama, Inoue, and Sato}{Sakama
  et~al\mbox{.}}{2017}]%
        {sakama2017linear}
\bibfield{author}{\bibinfo{person}{Chiaki Sakama}, \bibinfo{person}{Katsumi
  Inoue}, {and} \bibinfo{person}{Taisuke Sato}.}
  \bibinfo{year}{2017}\natexlab{}.
\newblock \showarticletitle{Linear algebraic characterization of logic
  programs}. In \bibinfo{booktitle}{\emph{Knowledge Science, Engineering and
  Management: 10th International Conference, KSEM 2017, Melbourne, VIC,
  Australia, August 19-20, 2017, Proceedings 10}}. Springer,
  \bibinfo{pages}{520--533}.
\newblock


\bibitem[\protect\citeauthoryear{Taha and Sheard}{Taha and Sheard}{2000}]%
        {TAHA2000211}
\bibfield{author}{\bibinfo{person}{Walid Taha} {and} \bibinfo{person}{Tim
  Sheard}.} \bibinfo{year}{2000}\natexlab{}.
\newblock \showarticletitle{MetaML and multi-stage programming with explicit
  annotations}.
\newblock \bibinfo{journal}{\emph{Theoretical Computer Science}}
  \bibinfo{volume}{248}, \bibinfo{number}{1} (\bibinfo{year}{2000}),
  \bibinfo{pages}{211--242}.
\newblock
\showISSN{0304-3975}
\urldef\tempurl%
\url{https://doi.org/10.1016/S0304-3975(00)00053-0}
\showDOI{\tempurl}
\newblock
\shownote{PEPM'97.}


\bibitem[\protect\citeauthoryear{Tomaz, Do~Nascimento, Hafid, and
  De~Souza}{Tomaz et~al\mbox{.}}{2020}]%
        {tomaz2020preserving}
\bibfield{author}{\bibinfo{person}{Antonio Emerson~Barros Tomaz},
  \bibinfo{person}{Jose~Claudio Do~Nascimento},
  \bibinfo{person}{Abdelhakim~Senhaji Hafid}, {and}
  \bibinfo{person}{Jose~Neuman De~Souza}.} \bibinfo{year}{2020}\natexlab{}.
\newblock \showarticletitle{Preserving privacy in mobile health systems using
  non-interactive zero-knowledge proof and blockchain}.
\newblock \bibinfo{journal}{\emph{IEEE access}}  \bibinfo{volume}{8}
  (\bibinfo{year}{2020}), \bibinfo{pages}{204441--204458}.
\newblock


\bibitem[\protect\citeauthoryear{Wood, Meent, and Mansinghka}{Wood
  et~al\mbox{.}}{2014}]%
        {wood2014new}
\bibfield{author}{\bibinfo{person}{Frank Wood}, \bibinfo{person}{Jan~Willem
  Meent}, {and} \bibinfo{person}{Vikash Mansinghka}.}
  \bibinfo{year}{2014}\natexlab{}.
\newblock \showarticletitle{A new approach to probabilistic programming
  inference}. In \bibinfo{booktitle}{\emph{Artificial intelligence and
  statistics}}. PMLR, \bibinfo{pages}{1024--1032}.
\newblock


\bibitem[\protect\citeauthoryear{Ye and Delaware}{Ye and Delaware}{2022}]%
        {10.1145/3498713}
\bibfield{author}{\bibinfo{person}{Qianchuan Ye} {and}
  \bibinfo{person}{Benjamin Delaware}.} \bibinfo{year}{2022}\natexlab{}.
\newblock \showarticletitle{Oblivious Algebraic Data Types}.
\newblock \bibinfo{journal}{\emph{Proc. ACM Program. Lang.}}
  \bibinfo{volume}{6}, \bibinfo{number}{POPL}, Article \bibinfo{articleno}{51}
  (\bibinfo{date}{jan} \bibinfo{year}{2022}), \bibinfo{numpages}{29}~pages.
\newblock
\urldef\tempurl%
\url{https://doi.org/10.1145/3498713}
\showDOI{\tempurl}


\bibitem[\protect\citeauthoryear{Zdancewic and Myers}{Zdancewic and
  Myers}{2001}]%
        {930133}
\bibfield{author}{\bibinfo{person}{S. Zdancewic} {and} \bibinfo{person}{A.C.
  Myers}.} \bibinfo{year}{2001}\natexlab{}.
\newblock \showarticletitle{Robust declassification}. In
  \bibinfo{booktitle}{\emph{Proceedings. 14th IEEE Computer Security
  Foundations Workshop, 2001.}} \bibinfo{pages}{15--23}.
\newblock
\urldef\tempurl%
\url{https://doi.org/10.1109/CSFW.2001.930133}
\showDOI{\tempurl}


\end{thebibliography}

\end{document}